\documentclass[a4paper,10pt,onecolumn,unpublished,nopdfoutputerror]{quantumarticle}
\usepackage{silence}
\usepackage[T1]{fontenc}
\usepackage[utf8]{inputenc}
\usepackage{amsmath,amssymb,amsthm,booktabs,array,longtable,graphicx}
\usepackage{tikz,microtype,pdflscape}
\usepackage[numbers,sort&compress]{natbib}
\usepackage{hyperref}
\newcommand{\ket}[1]{\lvert #1\rangle}
\newcommand{\bra}[1]{\langle #1\rvert}
\newcommand{\AME}{\operatorname{AME}}

\newcommand{\C}{\mathbb C}
\newcommand{\R}{\mathbb R}
\newcommand{\F}{\mathbb F}
\newcommand{\Q}{\mathbb Q}

\newtheorem{proposition}{Proposition}

\newtheorem{corollary}{Corollary}
\definecolor{col0}{HTML}{307AB4}
\definecolor{col1}{HTML}{D66A35}
\definecolor{col2}{HTML}{288F83}
\definecolor{col3}{HTML}{BE5475}
\definecolor{col4}{HTML}{8B72B3}
\definecolor{col5}{HTML}{A6AA37}
\begin{document}
\hypersetup{linkcolor=col0,citecolor=col0,urlcolor=col0}
\title{Fractional and quantum Latin squares:\newline
real solutions of quantum 36 officers problem}
\author{Mykhailo Hontarenko}
\affiliation{Institute of Theoretical Physics, Faculty of Physics, Astronomy and Applied Computer Science, Jagiellonian University, ul.\ \L ojasiewicza 11, 30-348 Krak\'ow, Poland}
\affiliation{Doctoral School of Exact and Natural Sciences, Jagiellonian University, ul.\ \L ojasiewicza 11, 30-348 Krak\'ow, Poland}
\author{Karol \.Zyczkowski}
\affiliation{Institute of Theoretical Physics, Faculty of Physics, Astronomy and Applied Computer Science, Jagiellonian University, ul.\ \L ojasiewicza 11, 30-348 Krak\'ow, Poland}
\affiliation{Center for Theoretical Physics, Polish Academy of Sciences, 
Aleja Lotnik\'ow 32/46, 02-668 Warsaw, Poland}
\date{October 12, 2026}
\maketitle
\hypersetup{pdftitle={Fractional and quantum Latin squares: real solutions of quantum 36 officers problem}}
\begin{abstract}
A Latin square can be extended by replacing its symbols with probabilities or quantum states. We relate these two extensions through explicit examples, from classical and fractional Latin squares to quantum Latin squares
and
a real solution of Euler's quantum 36 officers problem. A simultaneous cyclic-shift ansatz on a four-torus of tensor coefficients 
gives a real absolutely maximally entangled (AME) state of four six-level systems. A local orthogonal change of basis yields a representative with 96 nonzero amplitudes.
Two  single-party  arrays 
obtained from the corresponding quantum orthogonal Latin square yield
fractional Latin squares. Signs matter: comparison with 
the independent real construction of Bruzda and its sign variant distinguishes three local-unitary classes, two of which have identical probabilities.
We extend this probability representation to
existing  seven- and nine-party AME states with six levels each.
 Their amplitude matrices reveal 48 complex Hadamard blocks of order $3^3$ and six of order $6^4$, respectively. These constructions separate what can be read from a probability picture from the coherent conditions responsible for maximal entanglement.
\end{abstract}

\section{Introduction}

Euler's 36 officers problem asks for an arrangement of six regiments and six ranks in a square, with every regiment and every rank represented once in each row and column. 
Euler conjectured  \cite{Euler1782} and Tarry proved \cite{Tarry1900}
that such a configuration is not possible classically. Pairs of orthogonal Latin squares exist at every order $d\geq2$ except two and six \cite{Bose1960}. The 
situation changes if the problem becomes 'quantum' and each 
cell is allowed to contain a superposition. 
Such an approach allowed  Rather \emph{et al.}
to obtain a quantum orthogonal Latin square of order six -- a solution to the quantum version 
of the 36  officers problem, and at the same time, to construct an absolutely maximally entangled golden state of a four-dice system \cite{Rather2022,Zyczkowski2023}.

Subsequent complex constructions use biunimodular vectors~\cite{Rather2024Biunimodular}, two-unitary complex Hadamard matrices~\cite{BruzdaZyczkowski2024}, and analytic constructions from qutrit Clifford sectors~\cite{GrossGoedicke2025}. In particular, flat complex $\AME(4,6)$ representatives are already known. These results enlarge the class of quantum solutions, but do not 
 lead to a real solution. %

The above observation leaves room for two rather different extensions of a Latin square. In a fractional square, 
each cell can be divided into smaller figures
whose relative areas specify a
probability distribution over symbols.
In a quantum Latin square 
each cell corresponds to a normalized 
pure state, while each row and column
forms an orthonormal basis.
Squared amplitudes connect the two descriptions, but do not identify them. Probabilities determine
the relative area of each coloured region
in a given cell, while relative phases
of the superposition coefficients
determine whether the vectors represented by that picture are orthogonal.

The connection between maximally mixed reductions and quantum
error-correcting codes was developed by Scott~\cite{Scott2004}, who
related multipartite entanglement measures to the purity of reduced
states. Facchi \emph{et al.} formulated the search for maximally
multipartite entangled states through the average purity over
balanced bipartitions~\cite{Facchi2008} and studied the frustration
that arises when its minimum cannot be reached in every
cut~\cite{Facchi2010}. AME states attain this minimum simultaneously
and are therefore perfect maximally multipartite entangled states
in that terminology. 
Let us emphasize that $\AME(N,d)$ states do not exist for all pairs $(N,d)$: in particular, $\AME(4,2)$ and $\AME(7,2)$
do not exist~\cite{HiguchiSudbery2000,Huber2017}. 
Due to effective usage of techniques relying on artificial intelligence,
a significant progress in this field started in July 2026,
with several results
beginning to appear in August, as models Fable 5 and GPT-5.6 Sol
were released.

Our main construction is a real $\AME(4,6)$
state
(of $n=4$ systems with $d=6$ levels each),
invariant under a simultaneous cyclic shift of all four indices of the tensor. A suitable
local orthogonal transformation reduces its dense cyclic form to a sparse representative with 96 nonzero amplitudes. This gives a positive answer to problem~(T2) of Ref.~\cite{AMEReview2026}, which asks whether real coefficients suffice, or equivalently whether an orthogonal two-unitary matrix of order 36 exists. We give both bases and their connecting transformations, so the role of the symmetry can be followed independently of the numerical search.

The sparse basis is particularly convenient for a Latin-square description. Each quantum officer contains either two or four equally weighted product-basis terms. Their areas are halves or quarters, their two output labels provide two alphabets, and their signs retain the cancellations required for maximal entanglement. We compare this state with Bruzda's independent 60-term construction \cite{WB26} and with a tensor obtained by reversing one of its signs. An exact contraction distinguishes the three states even when permutations of the parties are allowed. The last two have the same probabilities, making the information lost by an area picture especially transparent.

The central result is the exact real $\AME(4,6)$ state  and its local transformations. The Latin-square viewpoint then gives a second result: fixed-position colour assembly can be possible in one local basis and impossible in another. In such a  way we establish a link between 
classical fractional Latin squares, quantum Latin squares
and quantum orthogonal Latin squares.

For more than four parties, cells become states of several output systems rather than individual qudits,
and one uses Latin cubes instead of squares.
We illustrate this passage using the  non-stabilizer states
$\AME(7,6)$  of Shi \emph{et al.}\ \cite{Shi2026} and $\AME(9,6)$  of Kitamura \cite{Kitamura2026}. The corresponding probability arrays are shown
to form a fractional Latin cube and a four-dimensional fractional Latin hypercube. Their pictures are very different: the seven-party example uses two alternating palettes, whereas the nine-party example has uniform probabilities and a nontrivial phase pattern. Grouping the amplitudes reveals 48 Hadamard blocks of order $27$ in the seven-party case, and six of order $6^4$ in the nine-party case.

This work is organized as follows. The next two sections recall fractional and quantum Latin squares. Section~\ref{sec:officers} explains their relation to quantum-officer arrays, including an order-four example giving $\AME(4,4)$. Section~\ref{sec:mh} presents the cyclic real construction and its sparse representative. Section~\ref{sec:invariant} compares three real four-party states, while Section~\ref{sec:higher} treats the seven- and nine-party examples. The appendices contain the local rotations, complete sparse coefficients, a 192-term basis with fixed-position colour sums, the sign witness and exact Hadamard-block proofs. Geomagic representations provide a related geometric viewpoint \cite{Sallows2011,WB26}; here our emphasis is on the Latin constraints and the amplitudes that lift them to quantum states.

\section{Latin squares and fractional entries}
\label{sec:fractional}
A Latin square of order $d$ contains $d$ symbols arranged in a $d\times d$ array, with each symbol appearing once in every row and column. Starting from the familiar cyclic example, $L_{ij}=i+j\pmod d$, one can express the same rule in a form suited to generalization: replace the symbols by a three-index incidence tensor,
\begin{equation}
 P^L_{ijk}=\delta_{k,L_{ij}},\qquad 0\leq i,j,k<d.
 \label{eq:latin-incidence}
\end{equation}
For fixed $i,j$, the third index records the symbol in a cell; the three indices thus encode the row, column and symbol on an equal footing. The Latin rule says that summing over any one of them, with the other two fixed, gives one.

Allowing a symbol to occupy only a fraction of a cell leads to a continuous version of this construction. In this work a \emph{fractional Latin square} is a real non-negative array $P$ satisfying
\begin{equation}
 \sum_k P_{ijk}=1,\qquad
 \sum_j P_{ijk}=1,\qquad
 \sum_i P_{ijk}=1.
 \label{eq:tristochastic}
\end{equation}
Such arrays are also called \emph{tristochastic} or three-dimensional line-stochastic tensors. We use the equal-weight case of fractional Latin squares and extend it to real entries, whereas the original definition in Ref.~\cite{Dugdale2000} allows rational symbol weights. Restricting Eq.~\eqref{eq:tristochastic} to $0/1$ entries recovers precisely the incidence tensors of ordinary Latin squares.

For a pictorial interpretation, assign a colour to each symbol and divide cell $(i,j)$ into regions of areas $P_{ij0},\ldots,P_{ij,d-1}$, taking the area of a cell to be one. The first condition then fills each cell, while the other two ensure that any given colour occupies total area one along a row or column. These equal areas follow from the Latin conditions. Whether the corresponding regions can also be assembled into a prescribed shape will depend on their geometry.

It is tempting to regard every fractional square as a random choice among classical squares. Indeed, any convex combination
\begin{equation}
 P=\sum_\alpha w_\alpha P^{L_\alpha},\qquad
 w_\alpha\geq0,\quad \sum_\alpha w_\alpha=1,
 \label{eq:latin-mixture}
\end{equation}
satisfies Eq.~\eqref{eq:tristochastic}. The converse, however, fails because the tristochastic polytope has vertices that are not Latin incidence tensors~\cite{Linial2014}; the Birkhoff theorem for doubly stochastic matrices therefore does not extend in this way. Throughout the paper, the ``tristochastic polytope'' denotes the full set defined by Eq.~\eqref{eq:tristochastic}, rather than only the convex hull of classical Latin squares.

\begin{proposition}[Fixed-position colour layouts]\label{prop:fixed-position}
Let $P$ be a fractional Latin square of order $d$. In each cell $(i,j)$, allow a measurable partition
$\{E_{ijk}\}_{k=0}^{d-1}$ of the same unit square $Q$, with
$|E_{ijk}|=P_{ijk}$. Regions may be disconnected, and all partition
conditions are understood up to sets of area zero. Require that,
for each colour $k$, the regions also partition $Q$ along every row
and every column, without moving them from their in-cell coordinates.
Such a layout exists if and only if $P$ is a convex combination of
Latin-square incidence tensors.
\end{proposition}
\begin{proof}
Away from a common set of area zero, each point $x\in Q$ belongs to
exactly one colour in each cell, and each colour occurs exactly once
in every row and column. These colours therefore form a Latin square
$L_x$. There are finitely many Latin squares of order $d$. Grouping the
points with the same $L_x$ and taking their areas gives nonnegative
weights $w_L$ with sum one and $P=\sum_Lw_LP^L$.
Conversely, partition $Q$ into vertical strips of widths $w_L$ and
colour strip $L$ in cell $(i,j)$ with the symbol $L_{ij}$. Every strip
satisfies the Latin rule, so all required partitions follow.
\end{proof}
The proposition concerns one selected colour alphabet. A simultaneous
rule for a second alphabet imposes further conditions. Nor does the
proposition require the region of an individual quantum coefficient
to be connected.

The same tensor defines a binary operation on probability distributions. Given distributions $p$ and $q$, their generalized $P$-convolution,
$r=p \star q$,
has components
\begin{equation}
 r_k=\sum_{i,j}P_{ijk}p_iq_j.
 \label{eq:classical-product}
\end{equation}
Positivity and the first line-sum condition make the result a probability distribution, while the other two imply that it is uniform whenever either argument is uniform. These operations, and their extensions to quantum states, were studied in Ref.~\cite{Bistron2023}. Here the classical tensor will organize the measurement probabilities of a quantum Latin square. To recover the square itself, however, we must also account for coherent orthogonality relations between the amplitudes.

\section{Quantum Latin squares and their probability arrays}
\label{sec:qls}
A quantum Latin square (QLS) of order $d$ is a $d\times d$ array of unit vectors $\ket{a_{ij}}\in\C^d$ whose rows and columns are orthonormal bases~\cite{Musto2016}:
\begin{equation}
 \langle a_{ij}|a_{ij'}\rangle=\delta_{jj'},\qquad
 \langle a_{ij}|a_{i'j}\rangle=\delta_{ii'}.
 \label{eq:qls}
\end{equation}
An ordinary Latin square becomes a QLS when each symbol is replaced by the corresponding basis vector, but Eq.~\eqref{eq:qls} also allows cells to contain superpositions. Once an output basis $\{\ket{k}\}$ is chosen, the resulting quantum square determines the probabilities
\begin{equation}
 P_{ijk}=|\langle k|a_{ij}\rangle|^2
 \label{eq:qls-probabilities}
\end{equation}
of its measurement outcomes,
which form a tristochastic tensor.

\begin{proposition}
The probability array of a quantum Latin square is tristochastic. If the square is a quantum Sudoku, its probabilities also satisfy the corresponding block constraints.
\end{proposition}
\begin{proof}
Normalization of $\ket{a_{ij}}$ gives $\sum_kP_{ijk}=1$. Completeness of each row gives $\sum_j\ket{a_{ij}}\bra{a_{ij}}=I_d$; taking the diagonal matrix element in $\ket{k}$ gives $\sum_jP_{ijk}=1$. Columns give the third equation. For a block $B$ of a quantum Sudoku \cite{Nechita2021,Paczos2021},
completeness similarly gives $\sum_{(i,j)\in B}P_{ijk}=1$.
\end{proof}

The diagonal matrix elements used in the proof retain the probabilities, but discard the phases. To see what is lost, write
$\langle k|a_{ij}\rangle=\sqrt{P_{ijk}}e^{i\phi_{ijk}}$. Row orthogonality then requires, for $j\ne j'$,
\begin{equation}
 \sum_k\sqrt{P_{ijk}P_{ij'k}}\,
 e^{i(\phi_{ij'k}-\phi_{ijk})}=0,
 \label{eq:phase-lift}
\end{equation}
with analogous equations in the columns. These cancellations impose conditions beyond the line sums of $P$. Consequently the probability array records part of the quantum square, without determining its coherent structure.

\subsection{A quantum Sudoku with sixteen distinct states}
Let $e_0,\ldots,e_3$ be an orthonormal basis. The following array illustrates the distinction between amplitudes and probabilities:
\begin{equation}
\renewcommand{\arraystretch}{1.22}
\begin{array}{c|c|c|c}
 e_0&e_1&e_2+e_3&e_2-e_3\\
 e_2&e_3&e_0-e_1&e_0+e_1\\\hline
 e_1+e_3&e_0-e_2&e_0+e_1+e_2-e_3&e_0-e_1+e_2+e_3\\
 e_1-e_3&e_0+e_2&e_0+e_1-e_2+e_3&e_0-e_1-e_2-e_3
\end{array}.
\label{eq:qls-four}
\end{equation}
With each entry normalized separately by a factor $1$, $1/\sqrt2$ or $1/2$, every row, column and indicated $2\times2$ block becomes an orthonormal basis. The array is therefore a quantum Sudoku, or SudoQ~\cite{Nechita2021,Paczos2021}, with sixteen distinct rays in place of the four used by a classical square. Since a common output unitary preserves the number of distinct rays, no such change of basis reduces this array to a classical Sudoku.
The pattern is genuinely quantum in the sense of Ref.~\cite{Paczos2021} and has the maximal cardinality sixteen. It also defines twelve distinct orthogonal measurements: four from the rows, four from the columns and four from the blocks. Its probability picture, however, contains only nine distinct distributions, since taking squared moduli removes the relative signs.

\begin{figure}[htbp]
\centering
\includegraphics[width=\linewidth]{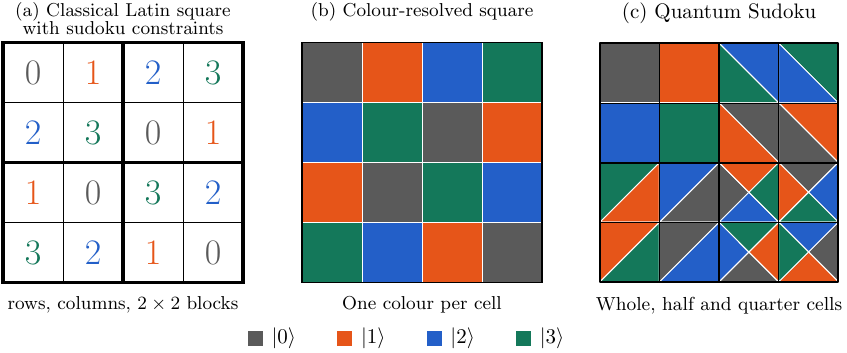}
\caption{From symbols to probability areas. (a) A classical Sudoku and (b) its colour encoding, with the $2\times2$ blocks marked by heavier lines. (c)  Probability array of the quantum Sudoku \eqref{eq:qls-four}
leads to a fractional Latin square of order four. Each region has area equal to the squared modulus of its coefficient. Taking square roots of these areas and restoring the signs in Eq.~\eqref{eq:qls-four} recovers the quantum states. Their sixteen distinct rays in Hilbert space give nine distinct probability distributions.}
\label{fig:karol-qls}
\end{figure}

The probability array nevertheless admits a simple classical realization. Consider the four ordinary Sudoku squares
\begin{equation}
\begin{aligned}
 S_1&=\begin{pmatrix}1&2&3&4\\3&4&1&2\\2&3&4&1\\4&1&2&3\end{pmatrix},&
 S_2&=\begin{pmatrix}1&2&3&4\\3&4&1&2\\4&1&2&3\\2&3&4&1\end{pmatrix},\\[4pt]
 S_3&=\begin{pmatrix}1&2&4&3\\3&4&2&1\\2&3&1&4\\4&1&3&2\end{pmatrix},&
 S_4&=\begin{pmatrix}1&2&4&3\\3&4&2&1\\4&1&3&2\\2&3&1&4\end{pmatrix}.
\end{aligned}
\label{eq:four-sudokus}
\end{equation}
Let $P^{S_\alpha}_{ijk}=\delta_{k,(S_\alpha)_{ij}-1}$, so that symbols $1,\ldots,4$ correspond to basis labels $0,\ldots,3$. Entry by entry, the square in Eq.~\eqref{eq:qls-four} satisfies
\begin{equation}
 P_{ijk}=\frac14\sum_{\alpha=1}^4P^{S_\alpha}_{ijk}.
 \label{eq:sudoq-decomposition}
\end{equation}
For example, the first cell always contains symbol~1; the third cell of the first row contains symbols~3 and~4 equally often; the third cell of the third row contains all four symbols once. These give precisely the probabilities $1$, $1/2$ and $1/4$ in the quantum array.

The classical realization in Eq.~\eqref{eq:sudoq-decomposition} reproduces the outcomes of a fixed basis measurement, although the cell vectors themselves remain coherent superpositions. Their relative signs enforce Eq.~\eqref{eq:phase-lift}, and replacing the vectors by incoherent mixtures would destroy this information. A related distinction will arise below for real AME states: identical squared amplitudes can conceal different entanglement classes, which a local-unitary invariant separates.

\section{From orthogonal squares to absolutely maximally entangled states}
\label{sec:officers}
\subsection{Two alphabets and quantum officers}
Two Latin squares $L$ and $M$ are orthogonal if their $d^2$ ordered pairs $(L_{ij},M_{ij})$ are all different. For instance, in $\F_3$,
\begin{equation}
 L_{ij}=i+j,\qquad M_{ij}=i+2j,\qquad
 \begin{array}{ccc}00&12&21\\11&20&02\\22&01&10\end{array}.
 \label{eq:ols-three}
\end{equation}
Here the map $(i,j)\mapsto(i+j,i+2j)$ has determinant one, so every ordered pair appears; fixing either input index still leaves each output symbol uniformly represented. At order six, however, classical symbols cannot satisfy these requirements simultaneously.

For two QLS arrays $a$ and $b$, orthogonality of their product cells means
\begin{equation}
 \langle a_{ij}|a_{i'j'}\rangle\langle b_{ij}|b_{i'j'}\rangle
 =\delta_{ii'}\delta_{jj'}.
 \label{eq:product-qols}
\end{equation}
In the quantum officers problem one may go further and allow the two symbols within a cell to be entangled \cite{Goyeneche2018,Rather2022}.
At order six, entanglement between the two output systems is necessary: Ball and Simoens ruled out an orthogonal pair of QLS, and hence a solution consisting entirely of product-state officers~\cite{BallSimoens2026}.
We therefore write a general cell as
\begin{equation}
 \ket{\Phi_{ij}}=\sum_{k,l=0}^{d-1}T_{ijkl}\ket{k,l}.
 \label{eq:cells}
\end{equation}
We reserve the name "orthogonal pair of QLS" 
for the product-cell construction and use ``quantum-officer array'' for this broader setting, also called a quantum orthogonal Latin square in the literature on Euler's problem. The distinction becomes transparent at the tensor level, where the requirement is that the three matrices
\begin{align}
 U_{ij,kl}&=T_{ijkl},&
 U^{R}_{ik,jl}&=T_{ijkl},&
 U^{\Gamma}_{il,kj}&=T_{ijkl}
 \label{eq:flattenings}
\end{align}
are all unitary. The last column order agrees with partial transposition on the second factor; exchanging its two column labels gives the equivalent $AD|BC$ flattening. Unitarity of the first matrix makes the $d^2$ officers mutually orthogonal, whereas the crossed conditions express the coherent row and column constraints.

For an orthogonal pair of QLS, the three equations follow from the two squares. In particular, substituting $T_{ijkl}=a_{ij,k}b_{ij,l}$ into a crossed Gram matrix gives
\begin{equation}
 \sum_{j,l}a_{ij,k}b_{ij,l}\overline{a_{i'j,k'}b_{i'j,l}},
\end{equation}
where the sum over $l$ vanishes for $i\ne i'$ by column orthogonality of $b$, and row completeness of $a$ gives $\delta_{kk'}$ when $i=i'$. Interchanging $a$ and $b$ proves the other crossed identity. The same matrix conditions remain meaningful for entangled cells even though this factorized argument is no longer available.

\subsection{Normalization and the three balanced cuts}
An $\AME(N,d)$ state is a pure state of $N$ parties with local
dimension $d$ whose reductions to any set of at most
$\lfloor N/2\rfloor$ parties are maximally mixed. Constructions
from classical maximum-distance-separable codes and applications
to quantum secret sharing are discussed by Helwig and
Cui~\cite{HelwigCui2013}.

\begin{proposition}\label{prop:ame}
Three different flattenings of the $4$-index tensor 
$T$ form unitary matrices of size $d^2$ if and only if
\begin{equation}
 \ket{\psi_T}=\frac1d\sum_{i,j,k,l=0}^{d-1}T_{ijkl}\ket{i,j,k,l}
 \label{eq:state}
\end{equation}
is an $\AME(4,d)$ state.
\end{proposition}
\begin{proof}
The coefficient matrix across $AB|CD$ is $U/d$. Its reduced density matrix is
\begin{equation}
 \rho_{AB}=\frac1{d^2}UU^\dagger.
 \label{eq:marginal}
\end{equation}
The other flattenings give $\rho_{AC}$ and $\rho_{AD}$ in the same way. Therefore, unitarity gives the maximally mixed marginals and the norm of the state. The complementary two-party marginals have the same spectra and are also maximally mixed. The converse follows from the same equations.
\end{proof}
These simultaneous conditions are restrictive: Higuchi and
Sudbery proved that no pure four-qubit state has all two-qubit
reductions maximally mixed, ruling out $\AME(4,2)$~\cite{HiguchiSudbery2000}.
Coding theory provides further restrictions. Rains introduced
quantum shadow enumerators to strengthen bounds on quantum
codes~\cite{Rains1999}; generalized shadow inequalities were later
used to obtain nonexistence bounds for AME states also in local
dimensions greater than two~\cite{Huber2018}.

The matrix $U$ is called two-unitary, and its four-leg tensor is perfect \cite{Pastawski2015,Goyeneche2015}. With real coefficients, three independent unitarity conditions reduce to orthogonality; for the classical pair in Eq.~\eqref{eq:ols-three}, the matrices are even permutations, giving the state
\begin{equation}
|\psi_{4,3}\rangle \ = \ 
 \frac13\sum_{i,j=0}^{2}\ket{i,j,i+j,i+2j},
\end{equation}
where addition is understood modulo $3$.
Quantum secret sharing provides one operational interpretation of these maximally mixed reductions 
\cite{Helwig2012}.

At the level of probabilities, $Q_{ijkl}=|T_{ijkl}|^2$ obeys six families of constraints: summing over any two indices gives one. As in the QLS case, these are the diagonal parts of the Gram identities, leaving their off-diagonal cancellations to the amplitudes. There is also a change in the marginal conditions: the three-index array of a single QLS has fixed line sums, whereas this four-index array has fixed sums over pairs of indices.

\subsection{A quantum square of sixteen officers}
An explicit example is the order-four array in Fig.~\ref{fig:karol-states}, where letters $\mathrm A,\mathrm B,\mathrm C,\mathrm D$ label one four-level system and numerals label the other. Its cells range from four product states through eight states of Schmidt rank two to four maximally entangled states of Schmidt rank four. The corresponding four-party state
\begin{equation}
 \ket{\psi_4}=\frac14\sum_{i,j=0}^{3}\ket{i,j}\otimes\ket{\Phi_{ij}}
 \label{eq:karol-four-state}
\end{equation}
has 36 nonzero amplitudes. The three flattenings of $4\psi_4$ have identity Gram matrices, or equivalently each two-party reduction of $\psi_4$ is $I_{16}/16$. The adjacent area picture displays both alphabets and their probabilities; to recover the orthogonality relations, one also needs the signs printed in the state square.

\begin{figure}[htbp]
\centering
\includegraphics[width=\linewidth]{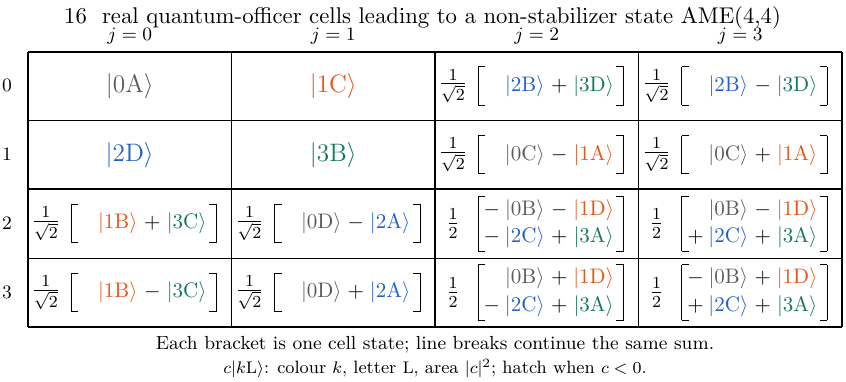}
\caption{Sixteen normalized quantum-officer cell states
inspired by the pattern in Fig.~\ref{fig:karol-qls}(c) and represented by probability areas in Fig.~\ref{fig:karol-qols}.
Numerals and letters label the two output systems. Every sign and normalization factor belongs to the tensor; the cells include product states and states of Schmidt ranks two and four.}
\label{fig:karol-states}
\end{figure}

\begin{figure}[htbp]
\centering
\includegraphics[width=\linewidth]{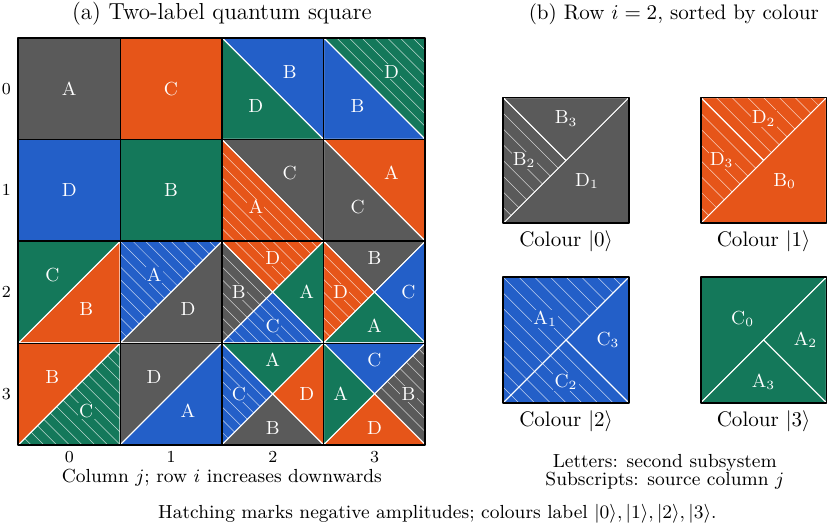}
\caption{An area representation of the order-four officer array. Its colour and letter marginals form two fractional Latin squares. Colours and letters denote its two alphabets; hatching marks negative coefficients. The four panels on the right regroup the cells of row $i=2$ by colour; the letter labels and hatching retain the second alphabet and coefficient signs.}
\label{fig:karol-qols}
\end{figure}

Although in the case of four-qutrit 
system all AME states are locally equivalent,
there exist several non-equivalent
 $AME(4,4)$ states \cite{Rather2023}.
It is easy to show that the
constructed state  (\ref{eq:karol-four-state})
is not locally equivalent to the following classical construction:
\begin{equation}
 \ket{\psi_{\rm cl}}=\frac14\sum_{x,y\in\F_4}\ket{x,y,x+y,x+\eta y},
 \qquad \eta^2+\eta+1=0.
 \label{AME44s}
\end{equation}
Indeed, the normalized Klein contraction defined in Section~\ref{sec:invariant}, with local index range four, gives $K(\psi_4)=13/4096$, whereas $K(\psi_{\rm cl})=1/256$. The two states are thus distinguished under complex local unitaries and party permutations, independently of how many terms appear in a particular basis.
Furthermore, in the finite-field Pauli convention over $\F_4$
(equivalently, two-qubit Pauli operators at each four-level site),
the state $|\psi_4\rangle$ is not a stabilizer state.
Its quantum magic, measured by the stabilizer R\'enyi-2 entropy
\cite{Leone2022}, is positive:
$M_2(\psi_4)=\ln(2048/291)\approx1.95130$,
using the natural logarithm, whereas the stabilizer state
in Eq.~\eqref{AME44s} has zero entropy in the same convention.

\subsection{The golden AME(4,6) state and the question  of a real one}
For the quantum 36 officers problem, the golden $\AME(4,6)$ state
constructed in Ref.~\cite{Rather2022} and subsequently analysed in
Refs.~\cite{Zyczkowski2023,AMEReview2026} is represented by a
two-unitary matrix with 112 nonzero entries. Their moduli are

\begin{equation}
 a=\sqrt{\frac{5-\sqrt5}{20}},\qquad
 b=\sqrt{\frac{5+\sqrt5}{20}},\qquad c=\frac1{\sqrt2},
 \qquad \frac ba=\frac{1+\sqrt5}{2}=\varphi.
\end{equation}
These are moduli in the normalized officer states. The corresponding amplitudes of the normalized four-party state have moduli $a/6$, $b/6$ and $c/6$.

Since the fraction $b/a$
is equal to the golden mean $\varphi$,
this AME state is called {\sl golden}. 
Suitable complex phases of all entries are given by powers of 
$\exp(i2\pi/20)$ and allow the state to satisfy the conditions for an AME state
and a quantum orthogonal Latin square of order six. The golden state solution, interpreted as a quantum superposition of four classical combinatorial designs called orthogonal frequency squares \cite{AMEReview2026}, relies on a suitable choice of complex phases. However, the natural question
posed as problem~(T2) of Ref.~\cite{AMEReview2026},
whether real coefficients suffice
to solve the problem,
 was left open.
The officer square in Appendix~\ref{app:golden-officers} displays its three coefficient moduli as areas, while retaining the relative phases explicitly. Restricting to a real tensor still allows phases $0$ and $\pi$ through the signs of its entries. As the next construction demonstrates, this freedom is enough to enforce all three orthogonality identities.

Under local unitaries the first flattening changes according to
\begin{equation}
 U'=(V_A\otimes V_B)U(V_C\otimes V_D)^{\mathsf T}.
 \label{eq:local-action}
\end{equation}
The transpose appears because both transformations act on ket indices. Although this action preserves all three unitary flattenings, it need not preserve the number of nonzero coefficients. We will use this freedom to pass between sparse, cyclic and geometrically convenient forms of the same state.

\section{Cyclic construction of a real 96-term
four-dice quantum state}
\label{sec:mh}
\subsection{A simultaneous cyclic shift}
Our construction began by restricting the symmetry of the tensor, while leaving its support free. The aim was to reduce the search space without imposing a particular pattern of zero coefficients. Let
\begin{equation}
 S\ket r=\ket{[r+1]_6},\qquad r=0,\ldots,5,
 \label{eq:mh-shift}
\end{equation}
where $[r]_6$ is the residue modulo six in $\{0,\ldots,5\}$. We seek a real state $\ket{\psi_{\rm cyc}}$ invariant under the simultaneous shift
\begin{equation}
 S^{\otimes4}\ket{\psi_{\rm cyc}}=\ket{\psi_{\rm cyc}}.
 \label{eq:mh-cyclic}
\end{equation}
Here all four indices move together, so ``diagonal'' refers to the group action rather than to a matrix diagonal. Each orbit contains six index tuples; specifying one coefficient per orbit reduces the number of real variables from $6^4$ to $6^3$.
The solution obtained within this ansatz has the exact form
\begin{equation}
 \ket{\psi_{\rm cyc}}=\frac1{6^4}
 \sum_{j,k,l=0}^5 c_{jkl}\sum_{s=0}^5
 \ket{s,[s+j]_6,[s+k]_6,[s+l]_6},
 \qquad c_{jkl}\in\mathbb Z[\sqrt2].
 \label{eq:mh-cyclic-state}
\end{equation}
For instance, $c_{000}=20+8\sqrt2$ and $c_{001}=-22+16\sqrt2$, and all $6^4$ amplitudes are nonzero in this basis.
For the numerical search, write
$T(x)_{ijkl}=x_{[j-i]_6,[k-i]_6,[l-i]_6}$ with $x\in\R^{6^3}$,
and let $U_\gamma(x)$ be its three balanced flattenings. The objective is
\begin{equation}
 \mathcal F(x)=\sum_{\gamma\in\{01|23,02|13,03|12\}}
 \left\|U_\gamma(x)U_\gamma(x)^{\mathsf T}-I_{36}\right\|_{\rm F}^{\,2}.
 \label{eq:cyclic-search-objective}
\end{equation}
This uses the officer normalization $T=6\psi$: at a zero of
$\mathcal F$, $\sum_{ijkl}T_{ijkl}^2=36$ and
$\sum_{jkl}x_{jkl}^2=6$. We minimized these Gram residuals by L-BFGS
with analytic gradients from random initial points. The candidate was
then refined by Newton iteration to 60 decimal digits. In the Fourier
basis, simultaneous-shift invariance restricts the support to
$m_0+m_1+m_2+m_3=0\pmod6$. Fixing eight local Fourier gauge parameters
and applying PSLQ identified coefficients in
$\Q(\sqrt2,\sqrt{-3})$; the inverse transform gives the real tensor
over $\Q(\sqrt2)$ above. The final Gram identities are exact and do
not depend on a numerical stopping tolerance.

\subsection{A local change of basis}
Although the cyclic form makes the symmetry explicit, it is unnecessarily long for specifying the state. Four real orthogonal changes of basis give a much shorter representative:
\begin{equation}
 \ket{\psi_{96}}=(O_1\otimes O_2\otimes O_3\otimes O_4)
 \ket{\psi_{\rm cyc}},\qquad O_aO_a^{\mathsf T}=I_6.
 \label{eq:mh-local-change}
\end{equation}
The matrices in Appendix~\ref{app:mh-rotations} were obtained by
minimizing the coefficient $\ell_1$ norm over four local orthogonal
matrices:
\begin{equation}
 \mathcal S(O_1,O_2,O_3,O_4)
 =\sum_{i,j,k,l}\left|
 \langle i,j,k,l|(O_1\otimes O_2\otimes O_3\otimes O_4)
 |\psi_{\rm cyc}\rangle\right|,
 \qquad O_a\in O(6).
 \label{eq:sparsification-objective}
\end{equation}
This optimization has 60 continuous local-basis parameters and leaves
the AME conditions intact. The resulting rotations and amplitudes
were reconstructed over $\Q(\sqrt2)$, then checked by exact
multiplication. With the printed matrices,
Eq.~\eqref{eq:mh-local-change} holds exactly; no thresholding of small
coefficients enters the final state.
In the first flattening, this transformation reads
\begin{equation}
 U_{96}=(O_1\otimes O_2)U_{\rm cyc}(O_3\otimes O_4)^{\mathsf T},
 \label{eq:mh-matrix-change}
\end{equation}
and the analogous transformations in the crossed cuts preserve the remaining two orthogonality conditions.

The simplification does not remove the cyclic symmetry, although that symmetry is no longer a shift of the displayed basis labels. Its generators are conjugated along with the state:
\begin{equation}
 R_a=O_aSO_a^{\mathsf T},\qquad R_a^6=I_6,\qquad
 (R_1\otimes R_2\otimes R_3\otimes R_4)\ket{\psi_{96}}
 =\ket{\psi_{96}}.
 \label{eq:mh-transformed-symmetry}
\end{equation}
In particular, $S^{\otimes4}\ket{\psi_{96}}\ne\ket{\psi_{96}}$ for the coefficients below. The two bases therefore reveal different features of the same state: the first displays the simultaneous shift, whereas the second displays 
the desired sparsity.

\subsection{Sparse real four-dice state }

In the sparse basis only two nonzero moduli remain. We denote the
corresponding tensor by $T^{(96)}$, with the normalization used in
Eq.~\eqref{eq:state}:
\begin{equation}
 \ket{\psi_{96}}=\frac16\sum_{i,j,k,l=0}^{5}
 T^{(96)}_{ijkl}\ket{i,j,k,l},
 \qquad
 T^{(96)}_{ijkl}\in
 \left\{0,\pm\frac12,\pm\frac1{\sqrt2}\right\}.
 \label{eq:mh}
\end{equation}
Thus $T^{(96)}_{ijkl}=6\langle i,j,k,l|\psi_{96}\rangle$.
Appendix~\ref{app:mh} lists the entries of $2T^{(96)}$; these belong to
$\{0,\pm1,\pm\sqrt2\}$ and require no additional tensor symbol.
In the remainder of this section, $T$ without a superscript denotes
$T^{(96)}$.
There are 48 entries of each nonzero modulus, so
\begin{equation}
 \langle\psi_{96}|\psi_{96}\rangle
 =\frac1{36}\left(\frac{48}{4}+\frac{48}{2}\right)=1.
\end{equation}
 The first flattening of $T^{(96)}$ has 24 rows with two entries of modulus $1/\sqrt2$, and 12 rows with four entries of modulus $1/2$; its columns have the same counts. The crossed flattenings instead have twelve two-entry rows and twenty-four three-entry rows, the latter with moduli $(1/\sqrt2,1/2,1/2)$.

\begin{proposition}\label{prop:mh}
All three balanced flattenings of $T^{(96)}$, defined in
Eq.~\eqref{eq:mh} and specified in Appendix~\ref{app:mh}, are orthogonal.
Consequently, $\ket{\psi_{96}}$ is a real $\AME(4,6)$ state.
\end{proposition}
\begin{proof}
For each balanced flattening $U$ of $T^{(96)}$, write
$2U=A+\sqrt2 B$, where $A$ and $B$ have integer entries. Then
\begin{equation}
 UU^{\mathsf T}=\frac14\left[
 AA^{\mathsf T}+2BB^{\mathsf T}
 +\sqrt2\left(AB^{\mathsf T}+BA^{\mathsf T}\right)\right].
\end{equation}
Since $1$ and $\sqrt2$ are linearly independent over $\mathbb Q$,
the identity $UU^{\mathsf T}=I_{36}$ is equivalent to
\begin{equation}
 AA^{\mathsf T}+2BB^{\mathsf T}=4I_{36},\qquad
 AB^{\mathsf T}+BA^{\mathsf T}=0.
 \label{eq:integer-gram}
\end{equation}
The coefficients in Appendix~\ref{app:mh} satisfy both identities in
all three cuts, including every off-diagonal entry. Both identities can be easily checked by integer matrix multiplication. The AME property follows from
Proposition~\ref{prop:ame}.
\end{proof}

\subsection{Thirty-six real quantum officers}
\label{sec:mh-officer-square}
The sparse representative admits a direct reading as a square of quantum
officers. Keep the first two indices as the row and column labels, and
regard the remaining two as the officer's two attributes:
\begin{equation}
 |\Phi_{ij}\rangle=\sum_{k,l=0}^{5}T^{(96)}_{ijkl}|k,l\rangle,
 \qquad
 |\psi_{96}\rangle=\frac16\sum_{i,j=0}^{5}|i,j\rangle\otimes|\Phi_{ij}\rangle.
 \label{eq:mh-officer-array}
\end{equation}
Here $T^{(96)}$ is the officer coefficient tensor defined in Eq.~\eqref{eq:mh}; the global amplitudes are $T^{(96)}_{ijkl}/6$. Numerals $0,\ldots,5$ label the first output system,
while letters $\mathrm A,\ldots,\mathrm F$ label the second. This change
of notation does not rotate or permute the tensor.

Within each cell the nonzero coefficients have equal moduli. Twenty-four
cells contain two terms with coefficients $\pm1/\sqrt2$, and twelve
contain four terms with coefficients $\pm1/2$. Moreover, no numeral or
letter is repeated within a cell. Its displayed expansion is therefore
already a Schmidt decomposition, up to the signs of the basis vectors.
The respective Schmidt ranks are two and four, with entanglement entropies
one and two bits. These ranks refer to the two output systems after fixing
the row and column; they are not ranks across a balanced cut of the
four-party AME state.

For example, the first and third cells of the first row are
\begin{align}
 |\Phi_{00}\rangle&=\frac{|1\mathrm C\rangle-|2\mathrm F\rangle}{\sqrt2},
 \\
 |\Phi_{02}\rangle&=\frac{-|1\mathrm B\rangle+|3\mathrm A\rangle
                         -|4\mathrm D\rangle+|5\mathrm F\rangle}{2}.
\end{align}
All 36 officers form an orthonormal basis of
$\mathbb C^6\otimes\mathbb C^6$. The two crossed orthogonality identities
are also satisfied, as follows from the exact integer Gram checks for
$T^{(96)}$. Thus the square retains the full real tensor, including the
relative signs responsible for coherent cancellation.

\begin{figure}[htbp]
 \centering
 \includegraphics[width=\linewidth]{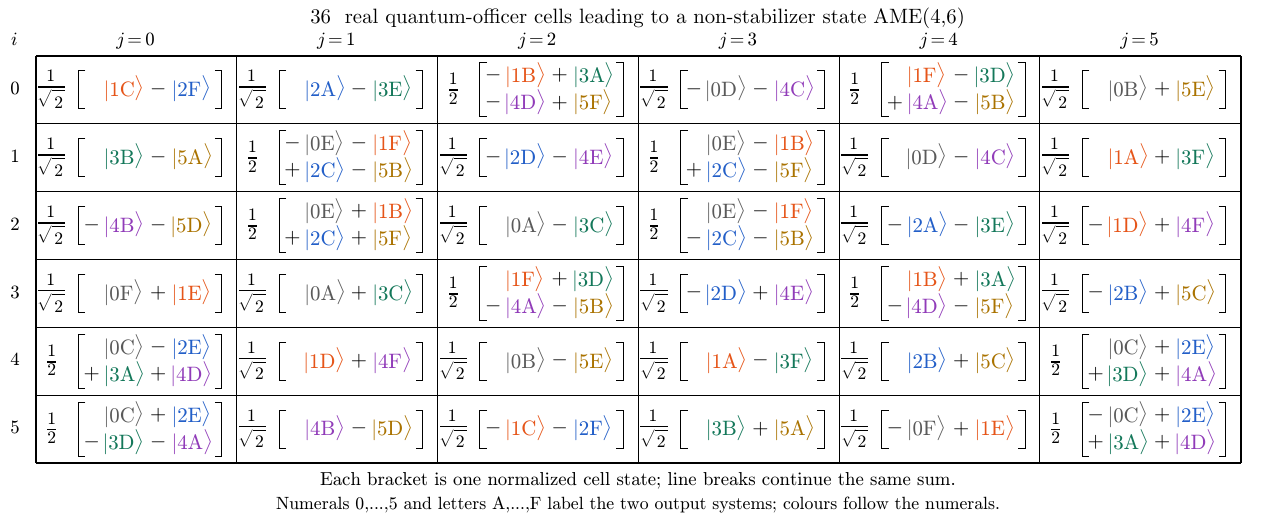}
 \caption{Thirty-six quantum-officer cell states of the real
 $\mathrm{AME}(4,6)$ state
 (Eq.~\eqref{eq:mh}) represented in Fig. \ref{fig:mh-officer-areas}. Numerals and letters label the two output systems;
 colour follows the numeral and does not encode a phase.
 The layout follows the order-four example: 24 cells contain two terms
 on a single line, while 12 contain four terms on two lines.
 All signs and normalization factors are retained. Each cell is
 maximally entangled on its two- or four-dimensional Schmidt support.
 For all coefficients without graphical encodings, see
 Table~\ref{tab:mh}.}
 \label{fig:mh-real-officers}
\end{figure}

\paragraph{Two fractional Latin squares.}
The two attributes give a useful probabilistic projection of the officer
array. Define
\begin{equation}
P_{ijk}=\sum_l|T^{(96)}_{ijkl}|^2,\qquad
 Q_{ijl}=\sum_k|T^{(96)}_{ijkl}|^2.
 \label{eq:mh-two-fractional}
\end{equation}
Both arrays have entries in $\{0,1/4,1/2\}$ and satisfy the three
line-sum conditions of a fractional Latin square. For example,
\begin{equation}
 \sum_kP_{ijk}=\sum_jP_{ijk}=\sum_iP_{ijk}=1,
\end{equation}
where each sum holds for every choice of the two remaining indices;
the analogous identities hold for $Q$.
For this representative, the one-party cell marginals are diagonal,
so these probabilities give their complete spectra and basis entries.

The two fractional squares should not be confused with a pair of pure
quantum Latin squares. Every officer in this displayed basis is
entangled, and hence cannot be written as
$|a_{ij}\rangle\otimes|b_{ij}\rangle$. Nor do the probabilities retain
all the information in the joint cell state. In $|\Phi_{00}\rangle$,
for instance, only the pairs $1\mathrm C$ and $2\mathrm F$ occur,
whereas multiplying its two marginal distributions would also produce
$1\mathrm F$ and $2\mathrm C$. The generalized quantum-officer array,
the two fractional Latin squares, and an ordinary product-cell pair of
QLS are therefore distinct objects.
\paragraph{A rational area representation.}
The two output alphabets can also be displayed within a single square,
as in the order-four example. Associate colour $k$ and letter
$\mathrm L_l$ with the term $T_{ijkl}|k,l\rangle$ in cell $(i,j)$.
Its region has area $|T_{ijkl}|^2$, with the area of a cell set to one;
hatching records a negative coefficient. For the present tensor, the
only nonzero areas are $1/2$ and $1/4$. The picture thus uses equal
halves in 24 cells and equal quarters in the remaining 12, while
retaining all 96 signs and ordered pairs of output labels.

The line sums of Eq.~\eqref{eq:mh-two-fractional} express the corresponding area balance. Consequently,
the fragments of each colour, collected along any row or column, have
total area one; the same is true for each letter. In this example they
can also be assembled into a unit square: their areas are either
$1/2+1/2$, $1/2+1/4+1/4$, or four quarters, and the triangular pieces
fit by rotations. Figure~\ref{fig:mh-officer-areas} illustrates this
assembly for the first row. This is a cut-and-reassemble representation,
not a claim of simultaneous tiling in the original in-cell positions.

The rationality concerns probabilities, not the amplitudes themselves,
which lie in $\mathbb Q(\sqrt2)$. With the hatching retained, the picture
specifies the full real tensor: a fragment of area $w$ represents the
coefficient $\pm\sqrt w$ in the normalized officer and
$\pm\sqrt w/6$ in the four-party state. Omitting the hatching would
discard the signs that enforce orthogonality. Thus the fractional
squares describe the probability balance, whereas the signed
quantum-officer array carries the coherent information required for AME.

\begin{figure}[htbp]
 \centering
 \includegraphics[width=\linewidth]{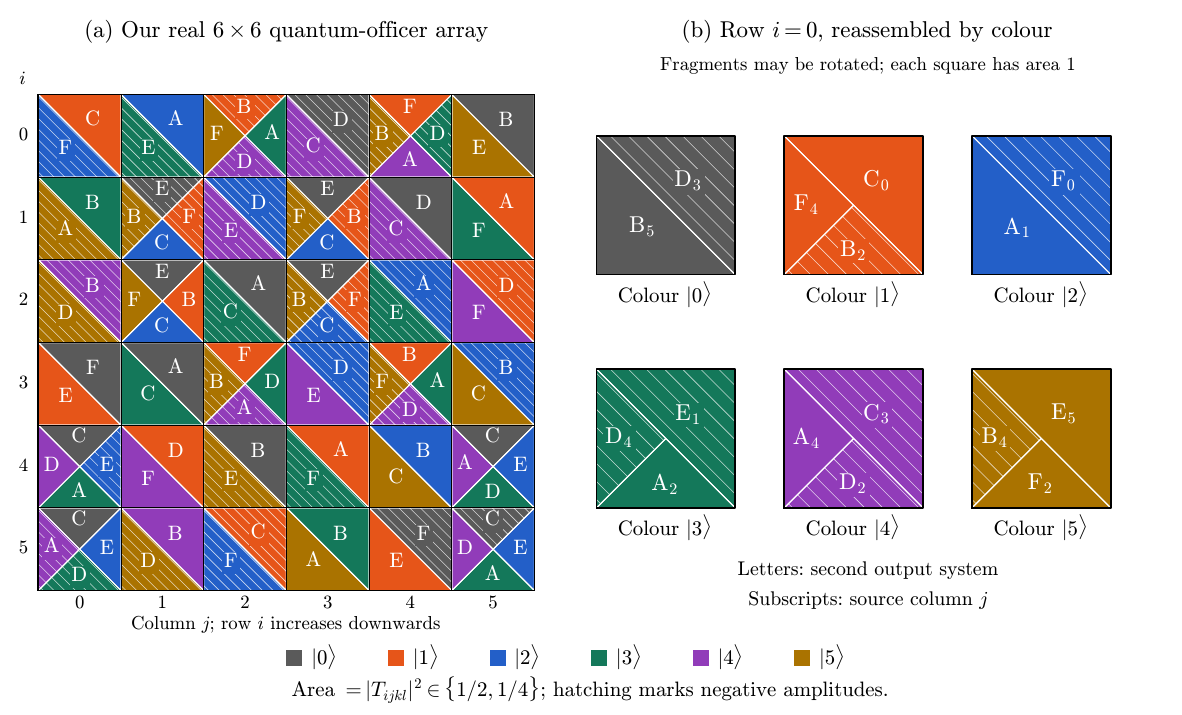}
 \caption{Geometric visualisation of real $\mathrm{AME}(4,6)$ 
 state (Eq.~\eqref{eq:mh})
 in its 96-term representation:
 (a) Each cell contains the normalized officer
 $|\Phi_{ij}\rangle=\sum_{k,l}T_{ijkl}|k,l\rangle$.
 Colour labels $k$, the letter labels $l$, area is $|T_{ijkl}|^2$,
 and hatching marks $T_{ijkl}<0$.
 (b) The fragments of row $i=0$, regrouped by colour and rotated as
 needed, form six unit squares. Letter subscripts give their source
 column $j$. The same area and reassembly statements hold for every
 row and column and for either alphabet; no change of quantum basis
 or tensor coefficient is involved. The colour and letter marginals
 define the two fractional Latin squares associated with this real
 AME state; hatching retains the coefficient signs.}
 \label{fig:mh-officer-areas}
\end{figure}

\begin{corollary}\label{cor:96-192-layout}
For the colour index $k$ in Eq.~\eqref{eq:mh-two-fractional}, the
displayed 96-term representative admits no fixed-position layout in
the sense of Proposition~\ref{prop:fixed-position}. The locally
equivalent 192-term representative in Eq.~\eqref{eq:mh192-state}
does admit such a layout.
\end{corollary}
\begin{proof}
Appendix~\ref{app:fixed-colour} enumerates the eight Latin squares
supported by the 96-term array and gives a linear identity satisfied
by every mixture of them, but violated by that array. For the
192-term array, Eq.~\eqref{eq:mh192-latin-mixture} gives an explicit
mixture of eight Latin squares. Proposition~\ref{prop:fixed-position}
then gives the two conclusions.
\end{proof}
This distinction concerns the selected colour alphabet and the
displayed bases, not the minimum support in the LU orbit. The
192-term layout does not simultaneously satisfy the analogous
letter rule. Its local rotation and colour assignment are given in
Appendix~\ref{app:fixed-colour}; Table~\ref{tab:mh192-officers} lists
all transformed officers. We retain the simpler 96-term picture here.

\section{Three real AME(4,6) states, two supports and one changing sign}
\label{sec:invariant}

The coefficients offer two rather different views of a state. Their
moduli determine the fractional Latin arrays and the coloured area
pictures, whereas the relative signs determine the cancellations needed
for orthogonality. A comparison with the construction of
Bruzda~\cite{WB26} makes this distinction particularly transparent.
Throughout this section, $T_{ijkl}=6\langle i,j,k,l|\psi\rangle$ for the normalized state under consideration. Our 96-term tensor is denoted by $T^{\mathrm{HZ}}=T^{(96)}$  and Bruzda's 60-term tensor by $T^{\rm WB}$.
The third tensor below is not an independent seed: it is obtained from
$T^{\rm WB}$ by reversing one coefficient. We call the original
representative \emph{Bruzda A} and this sign variant \emph{Bruzda B}.

\subsection{What the blocks do, and do not, tell us}

For each balanced cut, join a row label to a column label whenever the
corresponding matrix entry is nonzero. The connected components of this
bipartite graph determine the largest separation into blocks obtainable
by row and column permutations alone. For an orthogonal matrix every
component has the same number of row and column vertices.
Figure~\ref{fig:three-real-blocks} displays these components for the three
flattenings of each state.

In the $01|23$ cut, our tensor contains twelve $2\times2$ blocks and
one connected $12\times12$ block. Its other two cuts contain six
$2\times2$ and six $4\times4$ blocks. Bruzda's tensor instead has
twelve scalar blocks and twelve $2\times2$ blocks in every cut.
The $12\times12$ block cannot be split further by permutations alone;
this statement concerns the displayed support, not arbitrary local
changes of basis or the smallest support in the LU orbit.

The permutations in this figure only reorder the composite row and
column labels. They need not factor into permutations on the individual
particles and are therefore not being used as physical local operations.

\begin{figure}[tbp]
\centering
\includegraphics[width=\linewidth]{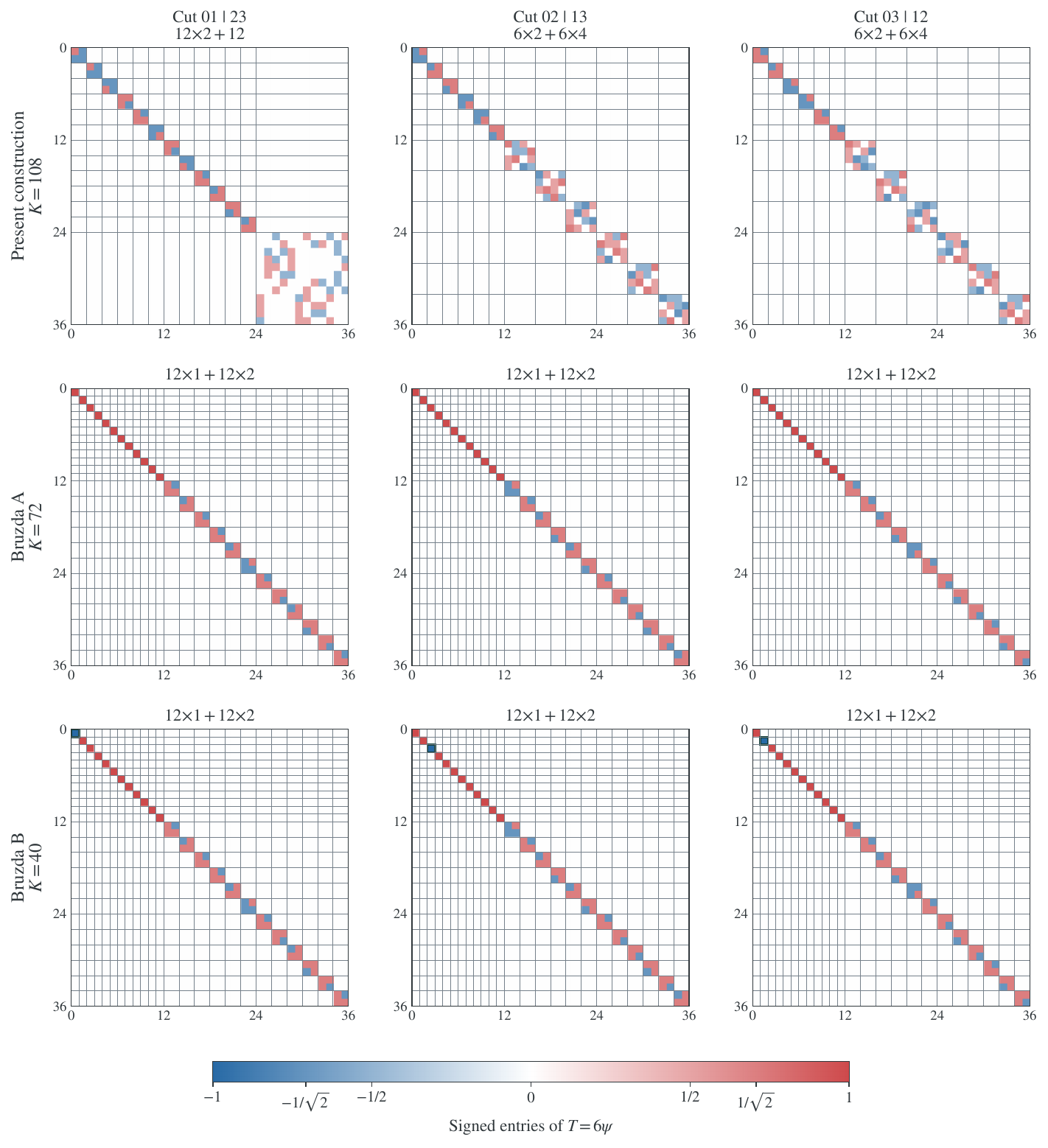}
\caption{The three balanced flattenings of three real AME$(4,6)$ states,
with rows and columns grouped by their exact support components.
Entries are those of $T=6\psi$; component sizes are written above each
panel. The lower two rows have identical support and differ only at
$T_{0132}$, outlined in dark green. The permutations are a display convention,
not an assertion of local equivalence. Blue and red encode negative and
positive entries, respectively; white denotes a zero coefficient.}
\label{fig:three-real-blocks}
\end{figure}

\subsection{An invariant separates the three states}

To distinguish entanglement classes we need a quantity unchanged by
local changes of basis. We use the degree-eight invariant
of Rather \emph{et al.}~\cite{Rather2023}, Eqs.~(19)--(20), with
the identity and the three nontrivial translations of $\F_2^2$ as
its four permutations. For four index tuples
$a,b,c,d\in\{0,\ldots,5\}^{4}$, it takes the tensor form
\begin{equation}
\begin{split}
 K(T)={} \sum_{a,b,c,d}T_aT_bT_cT_d\,
\overline{T_{a_0b_1c_2d_3}}\,
 \overline{T_{b_0a_1d_2c_3}}
  \overline{T_{c_0d_1a_2b_3}}\,
\overline{T_{d_0c_1b_2a_3}}.
\end{split}
\label{eq:klein}
\end{equation}
Every local index is paired between a tensor and a conjugate tensor, so
the local unitary matrices cancel. Thus $K$ is invariant under
$U(6)^{\otimes4}$. Its four matchings form the Klein group, which also
makes it invariant under permutations of the four particles; the
relabelling argument is given in Appendix~\ref{app:sign-witnesses}.
For normalized states, $K(\psi)=K(T)/6^8$.
The invariant and the use of probability-preserving
phase families are inherited from Ref.~\cite{Rather2023}. Here we
evaluate this invariant on the three real tensors, derive the
isolated-coefficient family below, and give the party-permutation
argument needed for this comparison.

In Bruzda's tensor, the entry $T^{\rm WB}_{0132}=1$ is the sole nonzero
entry of its row and column in each of the three flattenings. Reversing
its sign therefore preserves orthogonality:
\begin{equation}
 T^{{\rm WB},-}=T^{\rm WB}-2E_{0132},
 \qquad \psi^{{\rm WB},-}=T^{{\rm WB},-}/6,
 \label{eq:WB-sign-variant}
\end{equation}
where $E_{0132}$ is the unit coordinate tensor. This small change has a
consequence that the support picture alone does not reveal. Exact
evaluation gives
\begin{equation}
 K(T^{{\rm WB},-})=40,\qquad K(T^{\rm WB})=72,\qquad
 K(T^{\mathrm{HZ}})=108.
 \label{eq:three-seed-K}
\end{equation}
The three states are consequently pairwise inequivalent under complex
local unitaries, even if permutations of the particles are allowed.
Yet the two 60-term states have the same coefficient probabilities,
and hence the same fractional Latin arrays and any representation that
depends only on these probabilities. This is one reason to keep the signs
visible in the quantum-officer figures.

\subsection{The two real endpoints belong to a complex phase family}

The same isolated coefficient may acquire an arbitrary phase rather than
only a minus sign. The tensor
\begin{equation}
 T^{\rm WB}(\theta)=T^{\rm WB}+(e^{i\theta}-1)E_{0132}
 \label{eq:isolated-phase-family}
\end{equation}
remains two-unitary for every real $\theta$, since this phase never
enters an off-diagonal row inner product. Its invariant is
\begin{equation}
 K(T^{\rm WB}(\theta))=56+16\cos\theta.
 \label{eq:WB-phase-K}
\end{equation}
It varies strictly monotonically for $0\leq\theta\leq\pi$ and therefore
distinguishes a continuous family of complex LU classes with identical
coefficient probabilities. The displayed real endpoints are
$T^{\rm WB}$ and $T^{{\rm WB},-}$. The value 108 of our real tensor lies
outside the interval $[40,72]$, so no state in this particular family is
locally equivalent to it. This is not a classification of all real or
complex AME$(4,6)$ states.

\begin{figure}[tbp]
\centering
\includegraphics[width=\linewidth]{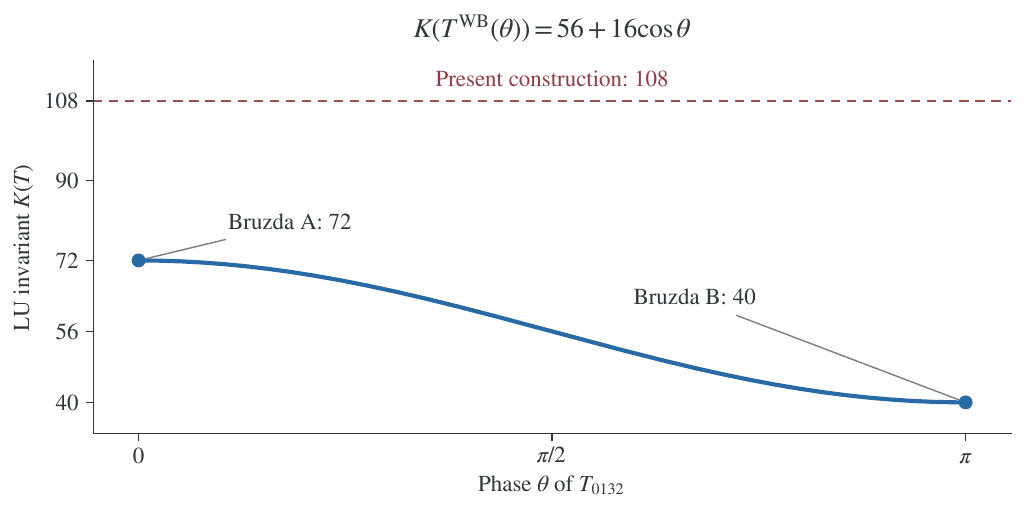}
\caption{The invariant varies along the isolated-coefficient phase
family although all coefficient probabilities remain fixed. Its two
real endpoints are Bruzda A state at $\theta=0$ and Bruzda B state
at $\theta=\pi$.
The dashed line marks $K(T^{\mathrm{HZ}})=108$, outside the interval covered
by this family. The normalization is $T=6\psi$.}
\label{fig:three-real-invariant}
\end{figure}

\subsection{Stabilizer R\'enyi entropy}
The three representatives can also be compared through their
nonstabilizerness, or quantum magic. We use the stabilizer R\'enyi
entropy~\cite{Leone2022} in its Weyl--Heisenberg form for
qudits~\cite{Klich2024}. For a normalized four-party state $|\psi\rangle$, set
\begin{equation}
 \begin{split}
 M_2(\psi)&=-\ln\!\left[
 \frac1{6^4}\sum_{\mathbf a,\mathbf b\in\mathbb Z_6^4}
 \left|\langle\psi|W_{\mathbf a,\mathbf b}|\psi\rangle\right|^4
 \right],\\
 W_{\mathbf a,\mathbf b}&=\bigotimes_{r=0}^{3}X^{a_r}Z^{b_r},\qquad
 X|j\rangle=|[j+1]_6\rangle,\quad
 Z|j\rangle=e^{2\pi ij/6}|j\rangle.
 \end{split}
 \label{eq:magic-M2}
\end{equation}
The logarithm is natural, so the values in
Table~\ref{tab:three-real-classes} are in nats. We evaluate all $6^8$
Weyl operators; for the three real representatives the fourth moments
are also certified exactly over $\Q(\sqrt2)$, as specified in
Appendix~\ref{app:sign-witnesses}.

The 96-term representative has the largest $M_2$ among the three.
For comparison, the original 112-term golden state~\cite{Rather2022}
gives $M_2=5.59179$, also below $5.89050$ for our representative.
This is a comparison of representatives in the stated cyclic Weyl convention, which is not a basis-independent ranking of the three entanglement classes.
$M_2$ is invariant
under Clifford transformations, but not under arbitrary local
unitaries. It therefore complements the invariant $K$ rather than
replacing it, and the ordering in the table is not an ordering of
LU classes by magic.

\begin{table}[tbp]
\centering
\small
\renewcommand{\arraystretch}{1.16}
\setlength{\tabcolsep}{6pt}
\begin{tabular}{@{}lrlrr@{}}
\toprule
Representative & Support & Moduli of $T$ (counts) & $K(T)$ & $M_2(\psi)$\\
 & & & & (nats)\\
\midrule
Present construction & 96 & $1/2$ (48), $1/\sqrt2$ (48) & 108 & 5.89\\
Bruzda A~\cite{WB26} & 60 & $1$ (12), \ \ \ $1/\sqrt2$ (48) & 72 & 5.58\\
Bruzda B, Eq.~\eqref{eq:WB-sign-variant} & 60
 & $1$ (12), \ \ \ $1/\sqrt2$ (48) & 40 & 5.57\\
\bottomrule
\end{tabular}
\par\smallskip
\begin{tabular}{@{}lcc@{}}
\toprule
Support components & Cut $01|23$ & Cuts $02|13$ and $03|12$\\
\midrule
Present construction & $12\,[2]+[12]$ & $6\,[2]+6\,[4]$\\
Bruzda A and B & $12\,[1]+12\,[2]$ & $12\,[1]+12\,[2]$\\
\bottomrule
\end{tabular}
\caption{Three real $\AME(4,6)$ representatives.
Moduli and $K$ refer to $T=6\psi$; $M_2$ refers to the normalized
state and the Weyl convention of Eq.~\eqref{eq:magic-M2}.
In the lower panel, $[m]$ denotes one connected $m\times m$ support
block after display permutations; the preceding integer is its
multiplicity. Support and $M_2$ depend on the displayed basis and have
not been optimized over the LU orbit. Different values of $K$ distinguish LU classes, whereas equal values alone do not establish equivalence.}
\label{tab:three-real-classes}
\end{table}

\section{From fractional Latin squares to fractional cubes and hypercubes}
\label{sec:higher}
The officer construction has a direct counterpart for an odd number of parties. The essential change is that the two sides of a largest balanced cut have different dimensions. For an $\AME(2m+1,d)$ state write
\begin{equation}
 |\Psi\rangle=d^{-m/2}\sum_{\mathbf a\in\{0,\ldots,d-1\}^m}
 |\mathbf a\rangle|\Phi_{\mathbf a}\rangle,
 \qquad
 T_{\mathbf a,\mathbf b}
 =d^{m/2}\langle\mathbf a,\mathbf b|\Psi\rangle.
 \label{eq:odd-officer}
\end{equation}
Here $T$ collects the coefficients of the normalized cell vectors:
$|\Phi_{\mathbf a}\rangle=\sum_{\mathbf b}T_{\mathbf a,\mathbf b}|\mathbf b\rangle$.
The $d^m$ cell states $|\Phi_{\mathbf a}\rangle$ are orthonormal vectors in $(\C^d)^{\otimes(m+1)}$. Thus $T$ is a $d^m\times d^{m+1}$ matrix with $TT^\dagger=I_{d^m}$; its adjoint is an isometry. It is not a square unitary matrix. This is the higher-party quantum-combinatorial setting discussed in Ref.~\cite{Goyeneche2018}, rather than an ordinary one-qudit quantum Latin cube.
Choose one of the output systems as a colour label $k$ and sum over the other $m$ outputs, denoted by $\mathbf c$:
\begin{equation}
 P_{a_1\cdots a_m k}=\sum_{\mathbf c}|T_{\mathbf a;k,\mathbf c}|^2.
 \label{eq:higher-probabilities}
\end{equation}
This array is a fractional Latin hypercube: summing over any one of its $m+1$ indices gives one, with the others fixed. Normalization of each cell proves the sum over $k$. For a sum over $a_r$, the remaining indices specify $m-1$ input systems and one output system. Their marginal in $|\Psi\rangle$ is $I/d^m$; multiplying its diagonal entries by $d^m$ proves the other line sums. The argument uses the AME property, not an assumed tiling.

The converse need not hold. Line sums fix diagonal marginal entries, whereas AME also fixes every off-diagonal entry. The examples below make this distinction visible, and show why one should keep a phase panel alongside the probability picture.

\subsection{A six-colour cube from the seven-dice state}
\label{sec:latin-cube}
For seven parties, the local dimension matters: it was proved that $\AME(7,2)$ does not exist~\cite{Huber2017}.
The six-level construction considered here is not subject to that
qubit-specific result.
We analyse the $\AME(7,6)$
state constructed in Ref.~\cite{Shi2026} and
specified in Appendix~\ref{app:seven-party}. Set $m=3$, label the cells by $(a,b,c)=(v_0,v_1,v_2)$ and select $k=v_5$ as the colour. The four output systems are $(v_3,v_4,v_5,v_6)$. In this basis every normalized cell has 162 nonzero amplitudes $\omega^r/\sqrt{162}$, where $\omega=\exp(2\pi i/3)$ and $r\in\{0,1,2\}$. The $6^3$ cell vectors form the rows of a $6^3\times6^4$ coefficient matrix.

The probability array has a particularly simple expression. If $\beta(v)=\lfloor v/3\rfloor$ is the binary part of a six-level label, then
\begin{equation}
 P_{abck}=\frac13\,\mathbf1_{\,\beta(k)=\beta(a)\oplus\beta(b)\oplus\beta(c)},
 \label{eq:cube-palettes}
\end{equation}
where $\oplus$ is addition modulo two. Each cell therefore contains either colours $0,1,2$ or colours $3,4,5$, in equal proportions. Along any of the $3\cdot6^2$ axial lines, each colour occurs in three cells with weight $1/3$. Its three strips assemble into a unit square. The repeated palettes identify probability distributions, not the underlying quantum cell states.

\begin{figure}[htbp]
 \centering
 \includegraphics[width=\linewidth]{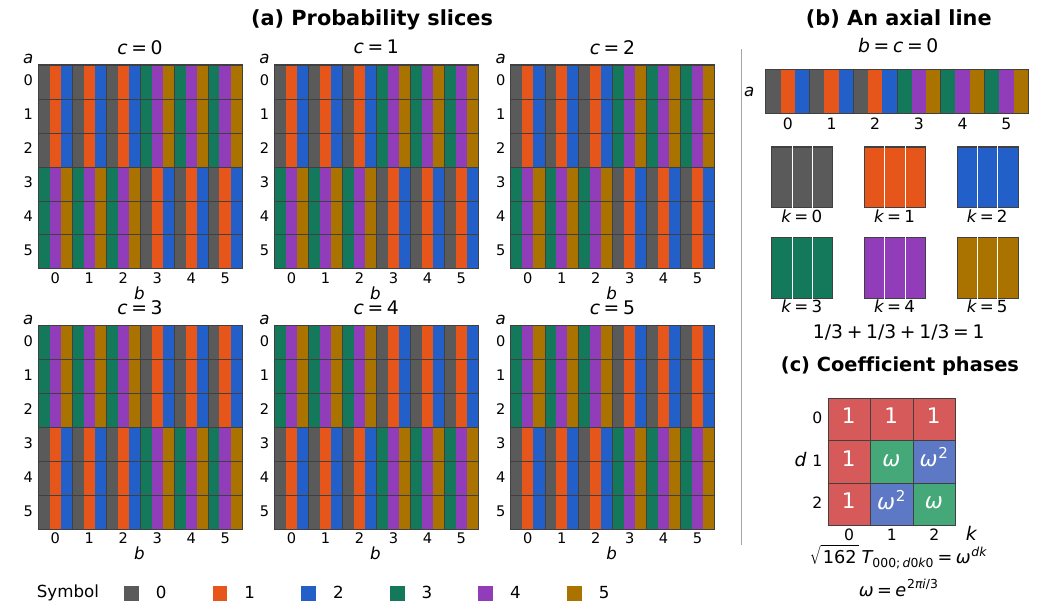}
 \caption{A fractional Latin cube obtained from the $\AME(7,6)$ state of Ref.~\cite{Shi2026}. (a) Six slices: rows $a$, columns $b$, layer $c$; colour is the outcome $k$ of the sixth particle. (b) An axial line contributes three one-third strips of each colour. (c) A phase-sensitive detail of cell $(0,0,0)$: for $e=f=0$ and $d,k=0,1,2$, the coefficients obey $\sqrt{162}\,T_{000;d0k0}=\omega^{dk}$. Colours label symbols in (a,b), but phases in (c). Equal areas alone do not reveal these phases.}
 \label{fig:ame76latin}
\end{figure}

Exact evaluation verifies all 35 three-party marginals, and direct counting verifies all $4\cdot6^3$ Latin line sums. Removing the phases leaves Eq.~\eqref{eq:cube-palettes} unchanged but destroys maximal entanglement. We use the cube to display this probability projection; no diagonal or simultaneous fixed-position tiling condition is part of the fractional Latin definition.

\subsection{Hadamard blocks behind the seven-party cube}
\label{sec:seven-hadamard}
The seven-party coefficient matrix contains 48 complex Hadamard blocks
of order 27. To locate them, write each site label as $v_i=3b_i+x_i$,
with $b_i\in\F_2$ and $x_i\in\F_3$, and let
$a=4b_0+2b_1+b_2\in\{0,\ldots,7\}$ label the input binary triple.
The output bits determine this triple and a branch
$\epsilon\in\{0,1\}$ of the source state. Within each group, order
the rows by $(x_0,x_1,x_2)$ and the columns by
$(\epsilon,x_4,x_3,x_5,x_6)$, all lexicographically.
The resulting display permutations $\Pi_r,\Pi_c$ give
\begin{equation}
 \Pi_r T\Pi_c=
 \bigoplus_{a=0}^{7}\frac1{\sqrt6}
 \begin{bmatrix}V_{a,0}&\cdots&V_{a,5}\end{bmatrix},
 \qquad V_{a,s}=\frac{H_{a,s}}{\sqrt{27}},
 \qquad H_{a,s}H_{a,s}^{\dagger}=27I_{27}.
 \label{eq:seven-hadamard-blocks}
\end{equation}
There are eight diagonal rectangles of size $27\times162$, each split
into six squares indexed by $s=3\epsilon+x_4\in\{0,\ldots,5\}$.
Using zero-based indices, the square $H_{a,s}/\sqrt{162}$ occupies
\begin{equation}
 \begin{split}
 \text{rows:}\quad&27a,\ldots,27a+26,\\
 \text{columns:}\quad&162a+27s,\ldots,162a+27s+26.
 \end{split}
 \label{eq:seven-hadamard-locations}
\end{equation}
These 48 disjoint squares contain all $48\cdot27^2=34992$ nonzero
entries of $T$; every other entry is zero. Each $H_{a,s}$ has entries
in $\{1,\omega,\omega^2\}$ and is therefore a Butson matrix
$\mathrm{BH}(27,3)$, whereas the corresponding coefficient block
$H_{a,s}/\sqrt{162}$ has entry modulus $1/\sqrt{162}$.
After removing row and column phases and permuting rows,
$H_{a,s}/\sqrt{27}$ becomes the unitary Fourier matrix
$F_3^{\otimes3}$. Appendix~\ref{app:seven-party} gives the common
character kernel.

\begin{figure}[htbp]
 \centering
 \includegraphics[width=\linewidth]{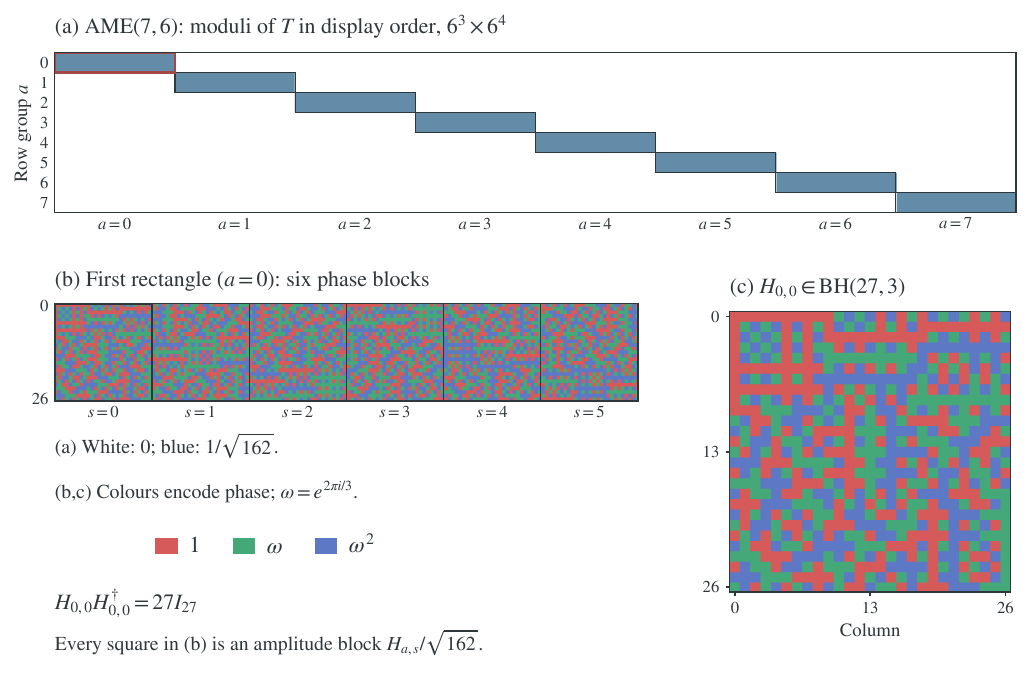}
 \caption{Hadamard structure of the $\AME(7,6)$ state of
 Ref.~\cite{Shi2026}, in the chosen three-versus-four flattening.
 (a) The complete coefficient matrix after display permutations, with
 eight $27\times162$ nonzero rectangles. (b) The rectangle $a=0$,
 divided into its six order-27 blocks. (c) The first square,
 $H_{0,0}$, from rows and columns $0,\ldots,26$.
 In (a), white denotes zero and blue denotes the common nonzero
 modulus $1/\sqrt{162}$. Colours in (b,c) show the phases
 $1,\omega,\omega^2$. The coefficient blocks of $T$ in (a) and (b) are
 $H_{a,s}/\sqrt{162}$, while (c) shows $H_{0,0}$.}
 \label{fig:seven-hadamard}
\end{figure}
\subsection{A flat nine-dice state and its phase hypercube}
\label{sec:latin-hypercube}
The nine-party state $\AME(9,6)$
was recently found by Kitamura \cite{Kitamura2026}. 
It has full support and can be written
\begin{equation}
 |\Psi_9\rangle=6^{-9/2}\sum_{\mathbf v\in\{0,\ldots,5\}^9}
 \zeta_6^{\,f(\mathbf v)}|\mathbf v\rangle,
 \qquad \zeta_6=e^{2\pi i/6},
 \label{eq:nine-phase-state}
\end{equation}
where $f$ is the integer-valued phase function modulo six specified in Eq.~\eqref{eq:ame96-phase}. The phase construction and source version are specified in Appendix~\ref{app:nine-phase}. Here ``flat'' means that all global amplitudes have modulus $6^{-9/2}$, not that they have the same phase.
Group the first four indices into $\mathbf a$ and the remaining five
into $\mathbf b$. Following Eq.~\eqref{eq:odd-officer}, define the
matrix of normalized cell coefficients by
\begin{equation}
 T_{\mathbf a,\mathbf b}
 :=6^2\langle\mathbf a,\mathbf b|\Psi_9\rangle
 =6^{-5/2}\zeta_6^{f(\mathbf a,\mathbf b)}.
 \label{eq:nine-coefficient-matrix}
\end{equation}
Here $f(\mathbf a,\mathbf b)$ is evaluated on the nine concatenated
local labels. The matrix $T$ has $6^4$ rows and $6^5$ columns.
Its rows specify the normalized cell states
\begin{equation}
 |\Phi_{\mathbf a}\rangle
 =\sum_{\mathbf b}T_{\mathbf a,\mathbf b}|\mathbf b\rangle,
 \qquad
 |\Psi_9\rangle=\frac1{6^2}
 \sum_{\mathbf a}|\mathbf a\rangle|\Phi_{\mathbf a}\rangle.
\end{equation}
The AME condition gives $TT^\dagger=I_{6^4}$, so the cell states
are orthonormal. The factor $6^2$ converts global amplitudes into
normalized cell coefficients; it does not identify a ket with
a matrix.

Every one-output probability in Eq.~\eqref{eq:higher-probabilities} is $1/6$. Consequently all $6^4$ cells have the same six-colour probability square. A probability picture alone is therefore uninformative about the phases. Figure~\ref{fig:ame96latin} shows the four-dimensional array through slices and includes actual amplitude phases to keep these two levels of description distinct.

\begin{figure}[htbp]
 \centering
 \includegraphics[width=\linewidth]{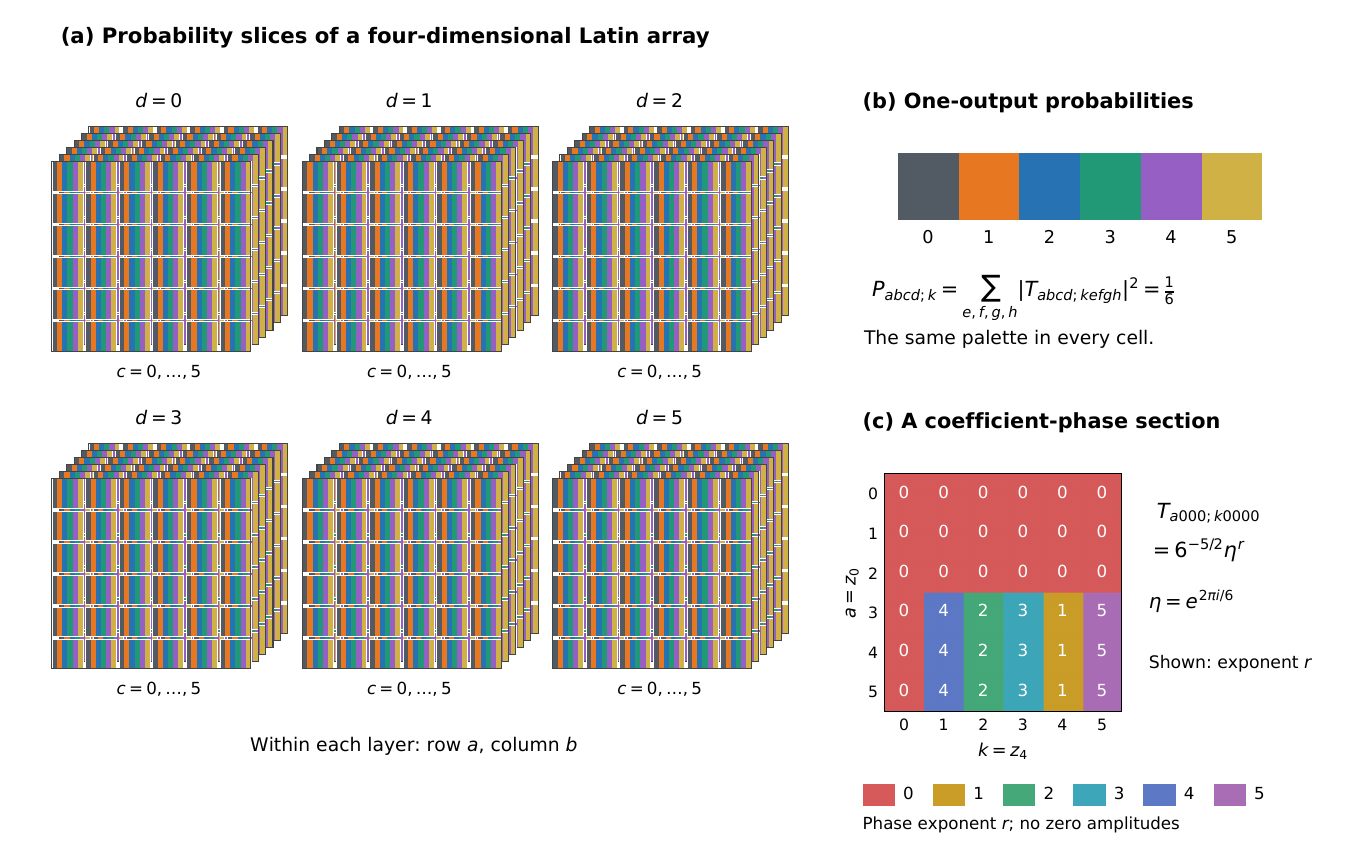}
 \caption{The fractional Latin hypercube associated with Eq.~\eqref{eq:nine-phase-state}, corresponding to $\AME(9,6)$. (a) Four input indices label the cells, displayed as six $d$-slices with $c$-layers. (b) The selected one-output marginal is uniform in every cell, so its six-colour square repeats. (c) A coefficient section: colours and numerals give $r$ in $T_{a000;k0000}=6^{-5/2}\zeta_6^r$. Here all undisplayed indices are fixed to zero, rather than traced out; $r=0$ means phase one, not a vanishing amplitude. Neither a repeated probability cell nor a phase slice identifies the full quantum state.}
 \label{fig:ame96latin}
\end{figure}

\subsection{Hadamard partitions of the amplitude matrix}
There is also a useful square structure inside the rectangular amplitude
matrix. Number the original parties from $0$ to $8$ and denote their local labels by $z_0,\ldots,z_8$. Reorder the output
columns so that $z_7=t$ is fixed first, followed by
$(z_4,z_5,z_6,z_8)$. Then
\begin{equation}
 TP_{\rm col}=\frac1{\sqrt6}
 \begin{bmatrix}V_0&V_1&\cdots&V_5\end{bmatrix},
 \qquad V_tV_t^\dagger=I_{6^4},\quad |(V_t)_{ab}|=6^{-2}.
 \label{eq:ame96-hadamard-blocks}
\end{equation}
Here $P_{\rm col}$ only records a rearrangement of the displayed column
labels. Each $6^2V_t$ is a Butson Hadamard matrix $\mathrm{BH}(6^4,6)$, with
sixth-root entries. In fact all six blocks have the same dephased form.
They are equivalent, under ordinary Hadamard row and column operations,
to the Fourier character table of
$\mathbb Z_2^4\times\mathbb Z_3^4$, or equivalently to
$F_6^{\otimes4}$ after normalization. The short proof is given in Appendix~\ref{app:nine-phase}.

The same array admits a finer partition into Hadamard matrices of order
36. Fix $z_0=a$, $z_1=b$, $z_4=c$, $z_6=d$ and $z_7=t$.
Use $(u,v)=(z_2,z_3)$ for the remaining row labels and
$(w,s)=(z_5,z_8)$ for the column labels. Define
\begin{equation}
 K^{a,b,c,d,t}_{(u,v),(w,s)}
 =6^{5/2}T_{(a,b,u,v),(c,w,d,t,s)},
 \qquad KK^\dagger=36I_{36}.
 \label{eq:ame96-submatrix36}
\end{equation}
All nine labels take values in $\{0,\ldots,5\}$. The factor $6^{5/2}$ removes the common modulus of the cell coefficients in Eq.~\eqref{eq:nine-coefficient-matrix}; equivalently, it multiplies each global amplitude by $6^{9/2}$. Thus $K$ has unimodular entries and $K/6$ is unitary.
For a fixed $t$, the $6^2$ choices of $(a,b)$ and $6^2$ choices of
$(c,d)$ give $6^4$ disjoint squares. Over all six values of $t$ they
give $6^5$ squares and account for every entry of the rectangular
matrix. The character-sum proof in Appendix~\ref{app:nine-phase}
applies to every choice of the five fixed labels.

Their locations can be read directly from the reordered matrix.
The zero-based row and column numbers of an entry of $K^{a,b,c,d,t}$ are
\begin{equation}
 r=216a+36b+6u+v,\qquad
 j=6^4t+216c+36w+6d+s.
 \label{eq:ame96-block-locations}
\end{equation}
Figure~\ref{fig:ame96-hadamard} shows $K^{0,0,0,0,0}$. Its rows are
$0,\ldots,35$, while the selected columns form six separated bands
$B_w$, numbered $36w,\ldots,36w+5$ for $w=0,\ldots,5$.
Retaining these bands from the displayed $36\times216$ rectangle gives
$K=[B_0\ \cdots\ B_5]$. The matching labels identify the same entries
before and after selection.

\begin{figure}[t]
 \centering
 \includegraphics[width=\linewidth]{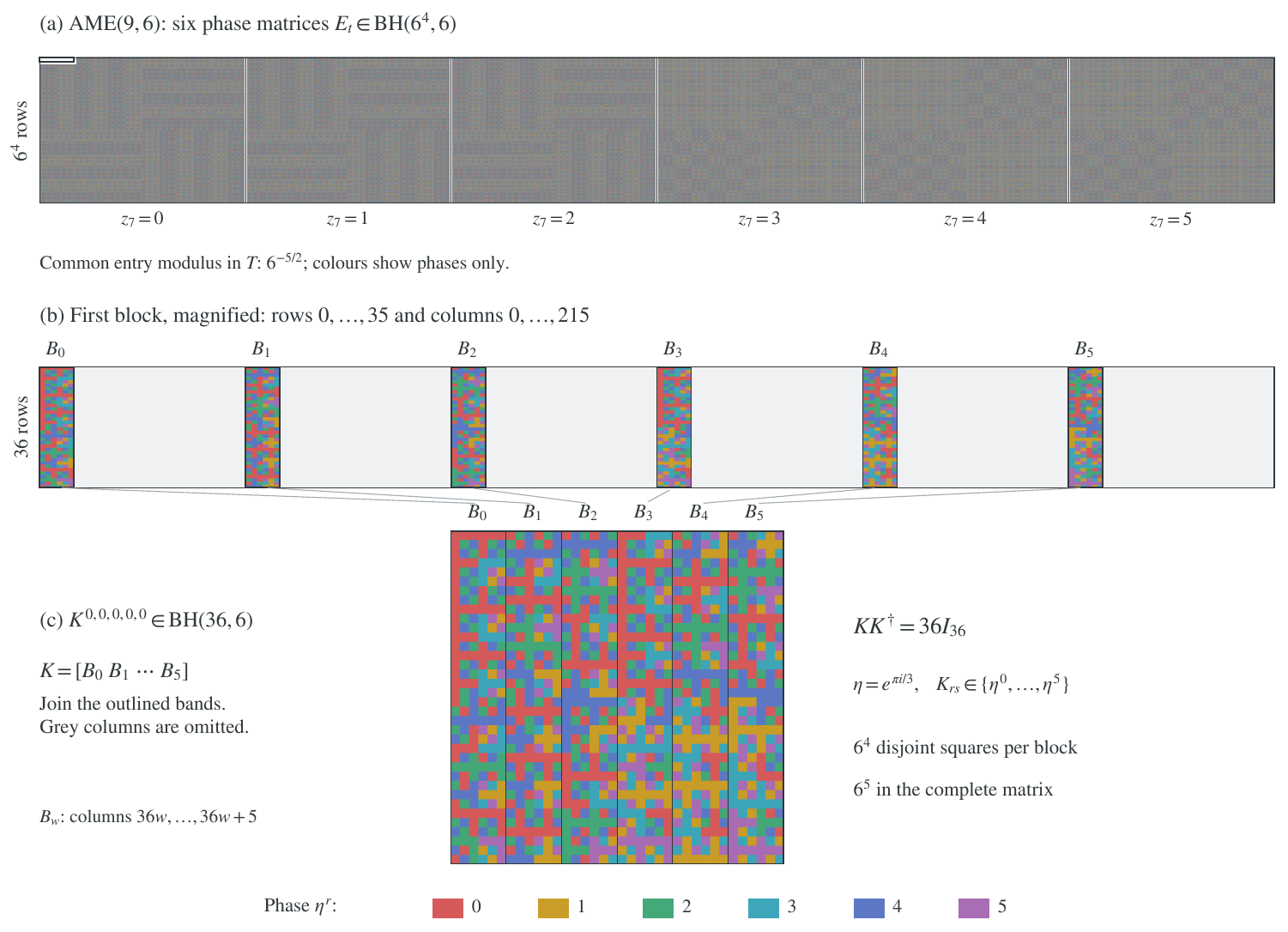}
 \caption{Hadamard submatrices of the $\mathrm{AME}(9,6)$ state of
 Ref.~\cite{Kitamura2026}. Colours denote sixth-root phases, with the
 common amplitude modulus suppressed.
 (a) Grouping columns by $z_7=t$ gives six blocks of order $6^4$.
 The frame marks the rectangle enlarged in (b).
 (b) Within this $36\times216$ rectangle, the outlined bands $B_w$
 contain columns $36w,\ldots,36w+5$; grey denotes omitted columns,
 not zero amplitudes.
 (c) Joining the six bands in order gives $K^{0,0,0,0,0}$ from
 Eq.~\eqref{eq:ame96-submatrix36}. The five fixed labels can be chosen
 independently, yielding a partition into $6^5$ squares of order 36.
 General locations are given in Eq.~\eqref{eq:ame96-block-locations}.}
 \label{fig:ame96-hadamard}
\end{figure}

These counts refer to the specified partitions of the amplitude matrix;
they are not a count of all possible Hadamard submatrices. The row and
column operations used for Hadamard equivalence act on composite
indices and need not be local operations on the individual particles.  A related connection appears in our recent stabilizer
construction~\cite{Hontarenko2026}: for every odd prime power $q$, the $\AME(2q,q)$ tensors admit a representation by normalized $q$-unitary complex Hadamard matrices. A forthcoming revision will also include an extension to stabilizer
$\AME(N,q)$ states for all $2\le N\le 2q$.

\section{Concluding remarks}
For each state considered here, squared amplitudes define fractional Latin arrays, and the figures represent these probabilities by areas. The probabilities alone, however, do not determine the state: relative phases are also needed to obtain the orthogonality conditions.

For the quantum 36 officers problem, a simultaneous cyclic-shift ansatz gives a real solution whose exact sparse representative has 96 terms. The two local bases serve different purposes. One reveals the symmetry that made the search practical; the other displays 36 entangled officers with only two coefficient moduli. Their rational areas provide a direct connection to fractional Latin squares, without requiring the coefficients themselves to be rational.

A further single-party basis change gives a 192-term representative whose colour regions cover the target square at fixed coordinates along every row and column. Its probability array is an explicit mixture of eight Latin squares, whereas the displayed 96-term array lies outside the convex hull of its supported Latin squares. The extra terms suffice for this construction; their necessity in the entire LU orbit is not established. The geometric subdivision is auxiliary and does not add quantum amplitudes.

Comparison with Bruzda's real solution and its one-sign variant distinguishes three entanglement classes. The two 60-term tensors have identical probabilities and support components, yet their invariant values differ. Conversely, a basis change may substantially alter a support picture without changing the class. Neither visible sparsity nor colour balance is a substitute for an LU invariant.

The higher-party examples extend the same distinction. In the seven-party cube, colours alternate between two palettes; in the nine-party hypercube every probability cell is uniform. In the chosen matrix cuts, $48$ Hadamard blocks of order $27$ cover the nonzero support of the seven-party state, whereas six blocks of order $6^4$ cover the nine-party array. The latter also admits a finer partition into $6^5$ squares of order 36, indexed by five fixed local labels. The character form makes their orthogonality transparent, although the conditions from the remaining cuts still enter the AME property.

Several questions remain open. The displayed sparse bases do not establish the minimum support in their local-unitary orbits. Flat complex $\AME(4,6)$ constructions are known~\cite{Rather2024Biunimodular,BruzdaZyczkowski2024,GrossGoedicke2025}; what remains open here is whether the specific LU class of $T^{\mathrm{HZ}}$ contains one. At the probabilistic level, it would be useful to characterize which rational fractional Latin arrays admit coherent lifts and which permit a fixed-position assembly for both output alphabets simultaneously. Proposition~\ref{prop:fixed-position} answers the one-alphabet question. Sign equations and finite-field phases give concrete settings in which to pursue these questions.

{
Let us conclude this work with two more general remarks.

\smallskip
A) Several multipartite states with extremal properties correspond to
quantum designs.
Classical probability vectors can be replaced by quantum
pure states, which can be 
used to construct a quantum superposition or convex mixture.
In an analogous way,
classical designs can be quantized, and their quantum superpositions
and mixtures can then be considered.

\smallskip
B) In the search for new quantum structures, including AME states,
a brute-force search in a space of very large dimension is unlikely to
succeed. We therefore suggest the following steps.

i) Make an educated guess about a condition that the desired state may satisfy,
in order to reduce the dimension of the search space.
In the case of a real $AME(4,6)$ state, this may be the shift symmetry
(\ref{eq:mh-cyclic})
on the 4--torus,
imposed on the coefficients of the tensor.

ii) With suitable computational tools (and some luck) you might be able to find a solution of this simplified problem.

iii) Exploit the gauge freedom of the numerical solution (for instance, local unitary
equivalence for AME states) to bring it to a special, transparent form,
which may make an analytic solution possible.
}

\paragraph{Acknowledgments.}
We are grateful to Wojciech Bruzda for sharing his independent real
construction \cite{WB26}, 
for providing materials, and for many discussions of the
geometric representation of quantum states. His input was essential to the development
of this work. We also thank Suhail Rather for sharing information and unpublished
notes on his independent construction~\cite{Ra26}, and for fruitful correspondence. His work will appear separately.
It is also a pleasure to thank Rafa{\l} Bistro{\'n},
Paolo Facchi, Arul Lakshminarayan,
Grzegorz Rajchel-Mieldzio{\'c}, Albert Rico
and Saverio Pascazio
for numerous discussions on quantum Latin squares and multipartite entanglement
which stimulated this research. We also thank Palash Pandya for 
timely correspondence.
M.H. and K.\,\!\.Z. acknowledge financial support from the European Union through ERC Advanced Grant \emph{TAtypic}, Project No.~101142236. M.H. also acknowledges support from the ``Research support module'' within the ``Excellence Initiative--Research University'' program at the Jagiellonian University in Krak\'ow.

\paragraph{Calculations.}
The appendices specify the 96-term tensor, its local transformations to the cyclic and 192-term forms, the full 192-term officer table and the one-sign comparison. 
The algebraic computations
use integer or exact finite-field arithmetic; the non-stabilizer entropy
of the golden state
was evaluated numerically.

\paragraph{Computational tools.}
AI tools assisted the numerical exploration, algebraic reconstruction, verification and manuscript preparation.
The authors are responsible for the content.

\appendix
\section{Local rotations and the cyclic representative}
\label{app:mh-rotations}
The following matrices act on the four parties in their order in the ket, with local basis order $\ket0,\ldots,\ket5$. They give the transformation in Eq.~\eqref{eq:mh-local-change}. Write $r=\sqrt2$ and take a common factor $1/6$ outside each matrix.

\begingroup
\setlength{\arraycolsep}{3.6pt}
\renewcommand{\arraystretch}{1.13}
\begin{equation}
\begin{gathered}
O_1=\frac16\begin{pmatrix}1&-5&-2&1&1&-2\\
2r&2r&-r&2r&2r&-r\\
3r&0&0&-3r&0&0\\
-1&-1&2&-1&5&2\\
0&0&-3r&0&0&3r\\
2r&-r&2r&2r&-r&2r\end{pmatrix}
\qquad
O_2=\frac16\begin{pmatrix}2r&-r&2r&2r&-r&2r\\
r&2r&-2r&-r&-2r&2r\\
1&3&4&-3&1&0\\
2r&2r&-r&2r&2r&-r\\
-3&1&0&1&3&4\\
-2r&2r&r&2r&-2r&-r\end{pmatrix}
\\[9pt]
O_3=\frac16\begin{pmatrix}3&0&3&-3&0&3\\
0&3r&0&0&-3r&0\\
-3&0&3&3&0&3\\
0&0&3r&0&0&-3r\\
0&3r&0&0&3r&0\\
3r&0&0&3r&0&0\end{pmatrix}
\qquad
O_4=\frac16\begin{pmatrix}r&-2r&-2r&r&-2r&-2r\\
r&2r&-2r&-r&-2r&2r\\
0&1&3&4&-3&1\\
-2r&-r&-2r&2r&r&2r\\
4&-3&1&0&1&3\\
-2r&-2r&r&-2r&-2r&r\end{pmatrix}
\end{gathered}
\label{eq:mh-rotations}
\end{equation}
\endgroup

The sparse table and these four matrices also specify the cyclic representative without an additional table of 216 numbers:
\begin{equation}
 \ket{\psi_{\rm cyc}}=\left(\bigotimes_{a=1}^4 O_a^{\mathsf T}\right)\ket{\psi_{96}},
 \qquad c_{jkl}=6^4\bra{0jkl}\psi_{\rm cyc}\rangle.
 \label{eq:mh-inverse}
\end{equation}
Equivalently, use the tensor $T^{(96)}$ specified by
Table~\ref{tab:mh}:
\begin{equation}
 c_{jkl}=216\sum_{a,b,u,v=0}^5
 (O_1)_{a0}(O_2)_{bj}(O_3)_{uk}(O_4)_{vl}\,
 T^{(96)}_{abuv}.
 \label{eq:mh-orbit-coefficients}
\end{equation}
The factor is $6^4/6=216$, since the global amplitudes are
$T^{(96)}_{abuv}/6$.
The result lies in $\mathbb Z[\sqrt2]$ for all 216 triples. Direct substitution checks $O_aO_a^{\mathsf T}=I_6$, the cyclic identity~\eqref{eq:mh-cyclic} and all three Gram identities. These checks require only integer arithmetic after the displayed denominators are cleared.

\section{The 96- and 192-term representatives of the real
\texorpdfstring{$\AME(4,6)$}{AME(4,6)} state}
\label{app:mh}

The two coefficient arrays below describe the same state in different
local bases. We first specify the sparse 96-term representative,
then give a single-party orthogonal transformation to a 192-term
representative whose colour regions admit fixed-position sums.

\subsection{Complete coefficients of the 96-term representative}

Table~\ref{tab:mh} specifies the state in Eq.~\eqref{eq:mh}
without a numerical approximation. The first two columns fix $(i,j)$;
each subsequent pair gives $(k,l)$ and the scaled coefficient
$2T^{(96)}_{ijkl}$. Divide a listed number by two to obtain an officer
coefficient, or by twelve to obtain a global state amplitude.
For example, the first entry gives $2T^{(96)}_{0012}=\sqrt2$,
hence $T^{(96)}_{0012}=1/\sqrt2$ and
$\langle0012|\psi_{96}\rangle=\sqrt2/12$.
The index pairs label two six-level systems; they are not decimal integers.

\begin{table}[htbp]
\centering
\caption{The 96 nonzero entries of $2T^{(96)}$ for the real
$\AME(4,6)$ state $\ket{\psi_{96}}$. Each listed number is twice
an officer coefficient and twelve times the corresponding amplitude
of the normalized four-party state. Blank slots, marked by a dash,
contain no further nonzero entries; all coefficients not listed are zero.}
\label{tab:mh}
\begingroup
\normalsize
\renewcommand{\arraystretch}{1.08}
\setlength{\tabcolsep}{4pt}
\begin{tabular*}{\textwidth}{@{\extracolsep{\fill}}cc cr cr cr cr@{}}
\toprule
\multicolumn{2}{c}{Row indices} & \multicolumn{8}{c}{Nonzero column indices and coefficients}\\
\cmidrule(lr){1-2}\cmidrule(lr){3-10}
$i$ & $j$ & $(k,l)$ & $2T^{(96)}_{ijkl}$ & $(k,l)$ & $2T^{(96)}_{ijkl}$ & $(k,l)$ & $2T^{(96)}_{ijkl}$ & $(k,l)$ & $2T^{(96)}_{ijkl}$\\
\midrule
$0$ & $0$ & $(1,2)$ & $+\sqrt2$ & $(2,5)$ & $-\sqrt2$ & \multicolumn{2}{c}{---} & \multicolumn{2}{c}{---}\\
$0$ & $1$ & $(2,0)$ & $+\sqrt2$ & $(3,4)$ & $-\sqrt2$ & \multicolumn{2}{c}{---} & \multicolumn{2}{c}{---}\\
$0$ & $2$ & $(1,1)$ & $-1$ & $(3,0)$ & $+1$ & $(4,3)$ & $-1$ & $(5,5)$ & $+1$\\
$0$ & $3$ & $(0,3)$ & $-\sqrt2$ & $(4,2)$ & $-\sqrt2$ & \multicolumn{2}{c}{---} & \multicolumn{2}{c}{---}\\
$0$ & $4$ & $(1,5)$ & $+1$ & $(3,3)$ & $-1$ & $(4,0)$ & $+1$ & $(5,1)$ & $-1$\\
$0$ & $5$ & $(0,1)$ & $+\sqrt2$ & $(5,4)$ & $+\sqrt2$ & \multicolumn{2}{c}{---} & \multicolumn{2}{c}{---}\\
\midrule
$1$ & $0$ & $(3,1)$ & $+\sqrt2$ & $(5,0)$ & $-\sqrt2$ & \multicolumn{2}{c}{---} & \multicolumn{2}{c}{---}\\
$1$ & $1$ & $(0,4)$ & $-1$ & $(1,5)$ & $-1$ & $(2,2)$ & $+1$ & $(5,1)$ & $-1$\\
$1$ & $2$ & $(2,3)$ & $-\sqrt2$ & $(4,4)$ & $-\sqrt2$ & \multicolumn{2}{c}{---} & \multicolumn{2}{c}{---}\\
$1$ & $3$ & $(0,4)$ & $+1$ & $(1,1)$ & $-1$ & $(2,2)$ & $+1$ & $(5,5)$ & $-1$\\
$1$ & $4$ & $(0,3)$ & $+\sqrt2$ & $(4,2)$ & $-\sqrt2$ & \multicolumn{2}{c}{---} & \multicolumn{2}{c}{---}\\
$1$ & $5$ & $(1,0)$ & $+\sqrt2$ & $(3,5)$ & $+\sqrt2$ & \multicolumn{2}{c}{---} & \multicolumn{2}{c}{---}\\
\midrule
$2$ & $0$ & $(4,1)$ & $-\sqrt2$ & $(5,3)$ & $-\sqrt2$ & \multicolumn{2}{c}{---} & \multicolumn{2}{c}{---}\\
$2$ & $1$ & $(0,4)$ & $+1$ & $(1,1)$ & $+1$ & $(2,2)$ & $+1$ & $(5,5)$ & $+1$\\
$2$ & $2$ & $(0,0)$ & $+\sqrt2$ & $(3,2)$ & $-\sqrt2$ & \multicolumn{2}{c}{---} & \multicolumn{2}{c}{---}\\
$2$ & $3$ & $(0,4)$ & $+1$ & $(1,5)$ & $-1$ & $(2,2)$ & $-1$ & $(5,1)$ & $-1$\\
$2$ & $4$ & $(2,0)$ & $-\sqrt2$ & $(3,4)$ & $-\sqrt2$ & \multicolumn{2}{c}{---} & \multicolumn{2}{c}{---}\\
$2$ & $5$ & $(1,3)$ & $-\sqrt2$ & $(4,5)$ & $+\sqrt2$ & \multicolumn{2}{c}{---} & \multicolumn{2}{c}{---}\\
\midrule
$3$ & $0$ & $(0,5)$ & $+\sqrt2$ & $(1,4)$ & $+\sqrt2$ & \multicolumn{2}{c}{---} & \multicolumn{2}{c}{---}\\
$3$ & $1$ & $(0,0)$ & $+\sqrt2$ & $(3,2)$ & $+\sqrt2$ & \multicolumn{2}{c}{---} & \multicolumn{2}{c}{---}\\
$3$ & $2$ & $(1,5)$ & $+1$ & $(3,3)$ & $+1$ & $(4,0)$ & $-1$ & $(5,1)$ & $-1$\\
$3$ & $3$ & $(2,3)$ & $-\sqrt2$ & $(4,4)$ & $+\sqrt2$ & \multicolumn{2}{c}{---} & \multicolumn{2}{c}{---}\\
$3$ & $4$ & $(1,1)$ & $+1$ & $(3,0)$ & $+1$ & $(4,3)$ & $-1$ & $(5,5)$ & $-1$\\
$3$ & $5$ & $(2,1)$ & $-\sqrt2$ & $(5,2)$ & $+\sqrt2$ & \multicolumn{2}{c}{---} & \multicolumn{2}{c}{---}\\
\midrule
$4$ & $0$ & $(0,2)$ & $+1$ & $(2,4)$ & $-1$ & $(3,0)$ & $+1$ & $(4,3)$ & $+1$\\
$4$ & $1$ & $(1,3)$ & $+\sqrt2$ & $(4,5)$ & $+\sqrt2$ & \multicolumn{2}{c}{---} & \multicolumn{2}{c}{---}\\
$4$ & $2$ & $(0,1)$ & $+\sqrt2$ & $(5,4)$ & $-\sqrt2$ & \multicolumn{2}{c}{---} & \multicolumn{2}{c}{---}\\
$4$ & $3$ & $(1,0)$ & $+\sqrt2$ & $(3,5)$ & $-\sqrt2$ & \multicolumn{2}{c}{---} & \multicolumn{2}{c}{---}\\
$4$ & $4$ & $(2,1)$ & $+\sqrt2$ & $(5,2)$ & $+\sqrt2$ & \multicolumn{2}{c}{---} & \multicolumn{2}{c}{---}\\
$4$ & $5$ & $(0,2)$ & $+1$ & $(2,4)$ & $+1$ & $(3,3)$ & $+1$ & $(4,0)$ & $+1$\\
\midrule
$5$ & $0$ & $(0,2)$ & $+1$ & $(2,4)$ & $+1$ & $(3,3)$ & $-1$ & $(4,0)$ & $-1$\\
$5$ & $1$ & $(4,1)$ & $+\sqrt2$ & $(5,3)$ & $-\sqrt2$ & \multicolumn{2}{c}{---} & \multicolumn{2}{c}{---}\\
$5$ & $2$ & $(1,2)$ & $-\sqrt2$ & $(2,5)$ & $-\sqrt2$ & \multicolumn{2}{c}{---} & \multicolumn{2}{c}{---}\\
$5$ & $3$ & $(3,1)$ & $+\sqrt2$ & $(5,0)$ & $+\sqrt2$ & \multicolumn{2}{c}{---} & \multicolumn{2}{c}{---}\\
$5$ & $4$ & $(0,5)$ & $-\sqrt2$ & $(1,4)$ & $+\sqrt2$ & \multicolumn{2}{c}{---} & \multicolumn{2}{c}{---}\\
$5$ & $5$ & $(0,2)$ & $-1$ & $(2,4)$ & $+1$ & $(3,0)$ & $+1$ & $(4,3)$ & $+1$\\
\bottomrule
\end{tabular*}
\endgroup
\end{table}

For an exact check, flatten the listed array $2T^{(96)}$ in each balanced cut and separate its coefficients of $1$ and $\sqrt2$ into integer matrices $A,B$. Their products satisfy Eq.~\eqref{eq:integer-gram}, verifying the normalization and the maximally mixed two-party marginals directly from the table.

\subsection{A local basis with fixed-position colour sums}
\label{app:fixed-colour}
The 96-term picture in Fig.~\ref{fig:mh-officer-areas} is deliberately
kept in the main text: its halves and quarters display the officers with
little geometric overhead. A different question is whether the pieces
of each colour can cover a square while retaining their original
positions inside the cells. The following local change of basis gives
such a representation of the same AME state.

\subsubsection{The orthogonal transformation and all coefficients}
With local basis order $0,\ldots,5$, set
\begin{equation}
 O=\frac1{\sqrt2}\begin{pmatrix}
 1&0&1&0&0&0\\
 0&1&0&0&1&0\\
 1&0&-1&0&0&0\\
 0&0&0&1&0&1\\
 0&1&0&0&-1&0\\
 0&0&0&1&0&-1
 \end{pmatrix},\qquad OO^{\mathsf T}=I_6,
 \label{eq:mh192-rotation}
\end{equation}
and define
\begin{equation}
 \ket{\psi_{192}}=(I\otimes I\otimes O\otimes I)\ket{\psi_{96}},
 \qquad T'_{ijul}=\sum_{k=0}^{5}O_{uk}T^{(96)}_{ijkl}.
 \label{eq:mh192-state}
\end{equation}
The matrix $O$ acts as the normalized real Hadamard matrix
$H_2=2^{-1/2}\bigl(\begin{smallmatrix}1&1\\1&-1\end{smallmatrix}\bigr)$
on the ordered pairs $(0,2)$, $(1,4)$ and $(3,5)$.
It acts on one physical party, unlike a permutation of a composite
row or column index. Thus Eq.~\eqref{eq:mh192-state} preserves the LU
class and every AME marginal. It does not put the whole state into a
flat Hadamard form.

The normalized officers $\ket{\Phi'_{ij}}=\sum_{k,l}T'_{ijkl}\ket{kl}$
are listed in Table~\ref{tab:mh192-officers}. Twenty-four cells have
four coefficients of modulus $1/2$; twelve have eight coefficients of
modulus $1/(2\sqrt2)$. In the global state the extra factor $1/6$ gives
96 amplitudes of modulus $1/12$ and 96 of modulus $1/(12\sqrt2)$.
For example,
\begin{equation}
 \ket{\Phi'_{00}}=\frac12\bigl(-\ket{0\mathrm F}
 +\ket{1\mathrm C}+\ket{2\mathrm F}+\ket{4\mathrm C}\bigr).
 \label{eq:mh192-cell00}
\end{equation}
More product-basis terms do not imply a larger Schmidt rank: the local
transformation leaves each cell's original rank, two or four, unchanged.
Combining Eq.~\eqref{eq:mh192-state} with Appendix~\ref{app:mh-rotations}
also gives the full path from the dense cyclic representative, with
third-party matrix $OO_3$.

To express the calculation in integer arithmetic, write $2T^{(96)}=a+b\sqrt2$, with integer coefficient arrays $a,b$, and set $R=\sqrt2 O$. Then
\begin{equation}
 T'=\frac{A+B\sqrt2}{4},\quad
 A_{ijul}=2\sum_kR_{uk}b_{ijkl},\quad
 B_{ijul}=\sum_kR_{uk}a_{ijkl}.
 \label{eq:mh192-integer}
\end{equation}
In each of the three square flattenings, exact integer multiplication gives
$AA^{\mathsf T}+2BB^{\mathsf T}=16I_{36}$ and
$AB^{\mathsf T}+BA^{\mathsf T}=0$.

\subsubsection{What is added in the colour picture}
Let $Q$ be a unit square and let $E_{ijkl}\subset Q$ be the region
assigned to the coefficient $T'_{ijkl}$ in cell $(i,j)$.
The regions partition each cell and have area $|T'_{ijkl}|^2$.
For a fixed colour $k$, their unions over $l$ partition $Q$ along
each row and each column, up to boundaries of zero area. No region is
moved, rotated or reflected when these unions are taken.
This is an addition of indicator functions and areas, not an addition
of quantum amplitudes.

The construction uses eight vertical strips of area $1/8$ at the same
positions in every cell. At position $r$, the colours form a Latin
square $L^{(r)}$. Consequently
\begin{equation}
 P'_{ijk}:=\sum_l|T'_{ijkl}|^2
 =\frac18\sum_{r=0}^{7}\delta_{k,L^{(r)}_{ij}}.
 \label{eq:mh192-latin-mixture}
\end{equation}
A weight $1/8$
occupies one strip; a weight $1/4$ occupies two. Neighbouring strips of
the same term are merged in Fig.~\ref{fig:mh192-fixed}. When its strips
are separated, matching circled tags identify them within that cell,
and its letter is printed only once. There are 192 amplitude labels;
67 terms have disconnected regions. In cell $(0,0)$, by contrast, the
four terms of Eq.~\eqref{eq:mh192-cell00} occupy four contiguous quarters.

\begin{figure}[p]
 \centering
 \includegraphics[width=\linewidth]{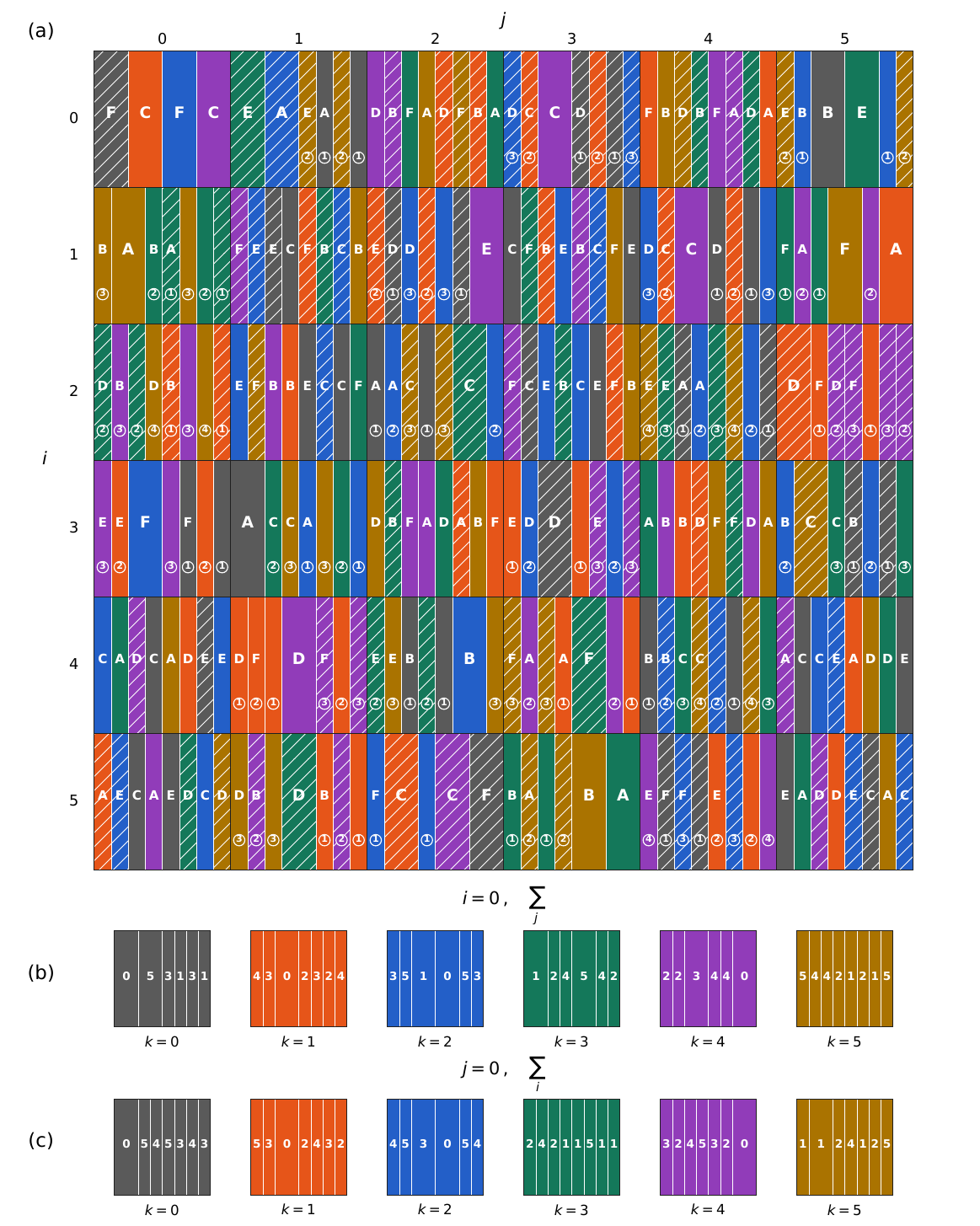}
 \caption{Fixed-position colour sums in the 192-term local basis.
 (a) The 36 normalized officers: colour labels $k$, letters
 $\mathrm A,\ldots,\mathrm F$ label $l=0,\ldots,5$, and hatching
 denotes a negative coefficient. Each term occupies total area
 $|T'_{ijkl}|^2=1/4$ or $1/8$. Matching circled numbers connect
 separated parts of the same term within one cell; its letter is
 printed only once. These tags are not additional quantum labels.
 (b) Colour regions from row $i=0$, overlaid at their original
 in-cell coordinates; ordinary digits identify the source column $j$.
 (c) The same operation for column $j=0$, with digits identifying
 the source row $i$. Each colour covers a unit square exactly once,
 without translating, rotating or reflecting any region. All six
 rows and six columns yield 72 such squares. The sums are of region
 indicators, rather than quantum amplitudes. The analogous simultaneous
 condition for the letter labels does not hold for this assignment.}
 \label{fig:mh192-fixed}
\end{figure}

\subsubsection{Why the 96-term representation is different}
By Proposition~\ref{prop:fixed-position}, a fixed-position
layout exists exactly when the colour array is a mixture of Latin
squares. A zero entry forbids the corresponding symbol in every
square with nonzero weight, so only Latin squares supported by the
array need be considered. Subdividing a colour region to carry the
second output label does not change this necessary colour condition.

For the 96-term tensor, complete exact enumeration gives eight Latin
squares compatible with the support of $P$ in
Eq.~\eqref{eq:mh-two-fractional}. Every one has
$L_{0,2}\in\{1,5\}$. A convex combination of them would therefore give
\begin{equation}
 \begin{aligned}
 &P_{0,2,1}+P_{0,2,5}=1
 &&\text{for every such convex combination},\\
 &P_{0,2,1}+P_{0,2,5}=\tfrac14+\tfrac14=\tfrac12
 &&\text{for the actual array}.
 \end{aligned}
 \label{eq:mh96-fixed-obstruction}
\end{equation}
This contradiction rules out a fixed-position layout for these cells and this colour index, even with arbitrary measurable subdivisions. It does not rule out reassembly by translations or rotations,
nor a different local basis. In particular, neither optimal support 192
nor the impossibility of every 96-term representative in the LU orbit
is asserted. Equally, we do not know whether disconnected regions are
necessary in every layout of the new basis. The geometry records a
choice of probability representation; the signed tensor remains the
object whose AME property is verified.

\begin{landscape}
\begin{table}[p]
\centering
\caption{The 36 officer cells of the 192-term representative of the real $\AME(4,6)$ state.
Row $i$, column $j$ gives $\ket{\Phi'_{ij}}$;
$\ket{\psi_{192}}=\frac16\sum_{i,j}\ket{ij}\ket{\Phi'_{ij}}$.
Digits and colours label the first output; letters
$\mathrm A,\ldots,\mathrm F$ label the second output $0,\ldots,5$.
Each ket is printed once, including its sign and coefficient. Geometric
subdivision of its probability does not duplicate it in this table.}
\label{tab:mh192-officers}
\resizebox{\linewidth}{!}{
\definecolor{c0}{HTML}{5A5A5A}
\definecolor{c1}{HTML}{E65519}
\definecolor{c2}{HTML}{235FC8}
\definecolor{c3}{HTML}{14785A}
\definecolor{c4}{HTML}{913CB9}
\definecolor{c5}{HTML}{AA7300}
\newcommand{\offket}[2]{{\color{c#1}|#1\mathrm{#2}\rangle}}
\begingroup
\setlength{\tabcolsep}{2.2pt}
\renewcommand{\arraystretch}{1}
\setlength{\jot}{0.5pt}
\fontsize{10}{11.5}\selectfont
\begin{tabular}{c|*{6}{c|}}
\multicolumn{1}{c}{}&\multicolumn{1}{c}{$j=0$}&\multicolumn{1}{c}{$j=1$}&\multicolumn{1}{c}{$j=2$}&\multicolumn{1}{c}{$j=3$}&\multicolumn{1}{c}{$j=4$}&\multicolumn{1}{c}{$j=5$}\\[2pt]\hline
\rule[-9.6mm]{0pt}{20mm}0&\makebox[41mm][c]{$\frac{1}{2}\left(\begin{gathered}-\offket{0}{F}+\offket{1}{C}\\+\offket{2}{F}+\offket{4}{C}\end{gathered}\right)$}&\makebox[41mm][c]{$\frac{1}{2}\left(\begin{gathered}\offket{0}{A}-\offket{2}{A}\\-\offket{3}{E}-\offket{5}{E}\end{gathered}\right)$}&\makebox[41mm][c]{$\frac{1}{2\sqrt{2}}\left(\begin{gathered}-\offket{1}{B}-\offket{1}{D}\\+\offket{3}{A}+\offket{3}{F}\\-\offket{4}{B}+\offket{4}{D}\\+\offket{5}{A}-\offket{5}{F}\end{gathered}\right)$}&\makebox[41mm][c]{$\frac{1}{2}\left(\begin{gathered}-\offket{0}{D}-\offket{1}{C}\\-\offket{2}{D}+\offket{4}{C}\end{gathered}\right)$}&\makebox[41mm][c]{$\frac{1}{2\sqrt{2}}\left(\begin{gathered}\offket{1}{A}+\offket{1}{F}\\-\offket{3}{B}-\offket{3}{D}\\-\offket{4}{A}+\offket{4}{F}\\+\offket{5}{B}-\offket{5}{D}\end{gathered}\right)$}&\makebox[41mm][c]{$\frac{1}{2}\left(\begin{gathered}\offket{0}{B}+\offket{2}{B}\\+\offket{3}{E}-\offket{5}{E}\end{gathered}\right)$}\\\hline
\rule[-9.6mm]{0pt}{20mm}1&\makebox[41mm][c]{$\frac{1}{2}\left(\begin{gathered}-\offket{3}{A}+\offket{3}{B}\\+\offket{5}{A}+\offket{5}{B}\end{gathered}\right)$}&\makebox[41mm][c]{$\frac{1}{2\sqrt{2}}\left(\begin{gathered}\offket{0}{C}-\offket{0}{E}\\-\offket{1}{F}-\offket{2}{C}\\-\offket{2}{E}-\offket{3}{B}\\-\offket{4}{F}+\offket{5}{B}\end{gathered}\right)$}&\makebox[41mm][c]{$\frac{1}{2}\left(\begin{gathered}-\offket{0}{D}-\offket{1}{E}\\+\offket{2}{D}+\offket{4}{E}\end{gathered}\right)$}&\makebox[41mm][c]{$\frac{1}{2\sqrt{2}}\left(\begin{gathered}\offket{0}{C}+\offket{0}{E}\\-\offket{1}{B}-\offket{2}{C}\\+\offket{2}{E}-\offket{3}{F}\\-\offket{4}{B}+\offket{5}{F}\end{gathered}\right)$}&\makebox[41mm][c]{$\frac{1}{2}\left(\begin{gathered}\offket{0}{D}-\offket{1}{C}\\+\offket{2}{D}+\offket{4}{C}\end{gathered}\right)$}&\makebox[41mm][c]{$\frac{1}{2}\left(\begin{gathered}\offket{1}{A}+\offket{3}{F}\\+\offket{4}{A}+\offket{5}{F}\end{gathered}\right)$}\\\hline
\rule[-9.6mm]{0pt}{20mm}2&\makebox[41mm][c]{$\frac{1}{2}\left(\begin{gathered}-\offket{1}{B}-\offket{3}{D}\\+\offket{4}{B}+\offket{5}{D}\end{gathered}\right)$}&\makebox[41mm][c]{$\frac{1}{2\sqrt{2}}\left(\begin{gathered}\offket{0}{C}+\offket{0}{E}\\+\offket{1}{B}-\offket{2}{C}\\+\offket{2}{E}+\offket{3}{F}\\+\offket{4}{B}-\offket{5}{F}\end{gathered}\right)$}&\makebox[41mm][c]{$\frac{1}{2}\left(\begin{gathered}\offket{0}{A}+\offket{2}{A}\\-\offket{3}{C}-\offket{5}{C}\end{gathered}\right)$}&\makebox[41mm][c]{$\frac{1}{2\sqrt{2}}\left(\begin{gathered}-\offket{0}{C}+\offket{0}{E}\\-\offket{1}{F}+\offket{2}{C}\\+\offket{2}{E}-\offket{3}{B}\\-\offket{4}{F}+\offket{5}{B}\end{gathered}\right)$}&\makebox[41mm][c]{$\frac{1}{2}\left(\begin{gathered}-\offket{0}{A}+\offket{2}{A}\\-\offket{3}{E}-\offket{5}{E}\end{gathered}\right)$}&\makebox[41mm][c]{$\frac{1}{2}\left(\begin{gathered}-\offket{1}{D}+\offket{1}{F}\\-\offket{4}{D}-\offket{4}{F}\end{gathered}\right)$}\\\hline
\rule[-9.6mm]{0pt}{20mm}3&\makebox[41mm][c]{$\frac{1}{2}\left(\begin{gathered}\offket{0}{F}+\offket{1}{E}\\+\offket{2}{F}+\offket{4}{E}\end{gathered}\right)$}&\makebox[41mm][c]{$\frac{1}{2}\left(\begin{gathered}\offket{0}{A}+\offket{2}{A}\\+\offket{3}{C}+\offket{5}{C}\end{gathered}\right)$}&\makebox[41mm][c]{$\frac{1}{2\sqrt{2}}\left(\begin{gathered}-\offket{1}{A}+\offket{1}{F}\\-\offket{3}{B}+\offket{3}{D}\\+\offket{4}{A}+\offket{4}{F}\\+\offket{5}{B}+\offket{5}{D}\end{gathered}\right)$}&\makebox[41mm][c]{$\frac{1}{2}\left(\begin{gathered}-\offket{0}{D}+\offket{1}{E}\\+\offket{2}{D}-\offket{4}{E}\end{gathered}\right)$}&\makebox[41mm][c]{$\frac{1}{2\sqrt{2}}\left(\begin{gathered}\offket{1}{B}-\offket{1}{D}\\+\offket{3}{A}-\offket{3}{F}\\+\offket{4}{B}+\offket{4}{D}\\+\offket{5}{A}+\offket{5}{F}\end{gathered}\right)$}&\makebox[41mm][c]{$\frac{1}{2}\left(\begin{gathered}-\offket{0}{B}+\offket{2}{B}\\+\offket{3}{C}-\offket{5}{C}\end{gathered}\right)$}\\\hline
\rule[-9.6mm]{0pt}{20mm}4&\makebox[41mm][c]{$\frac{1}{2\sqrt{2}}\left(\begin{gathered}\offket{0}{C}-\offket{0}{E}\\+\offket{1}{D}+\offket{2}{C}\\+\offket{2}{E}+\offket{3}{A}\\-\offket{4}{D}+\offket{5}{A}\end{gathered}\right)$}&\makebox[41mm][c]{$\frac{1}{2}\left(\begin{gathered}\offket{1}{D}+\offket{1}{F}\\+\offket{4}{D}-\offket{4}{F}\end{gathered}\right)$}&\makebox[41mm][c]{$\frac{1}{2}\left(\begin{gathered}\offket{0}{B}+\offket{2}{B}\\-\offket{3}{E}+\offket{5}{E}\end{gathered}\right)$}&\makebox[41mm][c]{$\frac{1}{2}\left(\begin{gathered}\offket{1}{A}-\offket{3}{F}\\+\offket{4}{A}-\offket{5}{F}\end{gathered}\right)$}&\makebox[41mm][c]{$\frac{1}{2}\left(\begin{gathered}\offket{0}{B}-\offket{2}{B}\\+\offket{3}{C}-\offket{5}{C}\end{gathered}\right)$}&\makebox[41mm][c]{$\frac{1}{2\sqrt{2}}\left(\begin{gathered}\offket{0}{C}+\offket{0}{E}\\+\offket{1}{A}+\offket{2}{C}\\-\offket{2}{E}+\offket{3}{D}\\-\offket{4}{A}+\offket{5}{D}\end{gathered}\right)$}\\\hline
\rule[-9.6mm]{0pt}{20mm}5&\makebox[41mm][c]{$\frac{1}{2\sqrt{2}}\left(\begin{gathered}\offket{0}{C}+\offket{0}{E}\\-\offket{1}{A}+\offket{2}{C}\\-\offket{2}{E}-\offket{3}{D}\\+\offket{4}{A}-\offket{5}{D}\end{gathered}\right)$}&\makebox[41mm][c]{$\frac{1}{2}\left(\begin{gathered}\offket{1}{B}-\offket{3}{D}\\-\offket{4}{B}+\offket{5}{D}\end{gathered}\right)$}&\makebox[41mm][c]{$\frac{1}{2}\left(\begin{gathered}-\offket{0}{F}-\offket{1}{C}\\+\offket{2}{F}-\offket{4}{C}\end{gathered}\right)$}&\makebox[41mm][c]{$\frac{1}{2}\left(\begin{gathered}\offket{3}{A}+\offket{3}{B}\\-\offket{5}{A}+\offket{5}{B}\end{gathered}\right)$}&\makebox[41mm][c]{$\frac{1}{2}\left(\begin{gathered}-\offket{0}{F}+\offket{1}{E}\\-\offket{2}{F}+\offket{4}{E}\end{gathered}\right)$}&\makebox[41mm][c]{$\frac{1}{2\sqrt{2}}\left(\begin{gathered}-\offket{0}{C}+\offket{0}{E}\\+\offket{1}{D}-\offket{2}{C}\\-\offket{2}{E}+\offket{3}{A}\\-\offket{4}{D}+\offket{5}{A}\end{gathered}\right)$}\\\hline
\end{tabular}
\endgroup
}
\end{table}
\end{landscape}

\section{Integer identities and the one-sign construction}
\label{app:sign-witnesses}

The seed of Bruzda A is taken from Bruzda's author-supplied draft~\cite{WB26}. Reconstruction from the signed Pauli table
in Eq.~(2) of that draft and the odd permutations of $(0,1,2,3)$
gives the same 60 coefficients. Indices are zero-based and ordered
$(i,j,k,l)$, with composite row and column indices $6i+j$ and $6k+l$
for the first flattening. Each nonzero coefficient can be written
using an integer $c$ and $m\in\{0,1\}$ as
\begin{equation}
 \psi_{ijkl}=c(\sqrt2)^m/24.
\end{equation}
Equivalently, write $T=6\psi=(A+\sqrt2 B)/4$, where $A$ and $B$ have
integer entries and disjoint supports. For each of the three balanced
flattenings, orthogonality is checked by the two integer identities
\begin{equation}
 AA^{\mathsf T}+2BB^{\mathsf T}=16I_{36},\qquad
 AB^{\mathsf T}+BA^{\mathsf T}=0.
 \label{eq:integer-Gram-checks}
\end{equation}
The three tensors satisfy these identities in all three cuts. Consequently $\psi=T/6$ is normalized and
every two-party reduced density matrix is $I_{36}/36$. The third tensor
is obtained from Bruzda's tensor by the single sign flip
$\psi_{0132}:1/6\mapsto-1/6$; no other coefficient is changed.

For completeness, the general fixed-support sign test is elementary.
In each flattening of the 60-term support, distinct rows meet in either
zero or two columns. When they meet in columns $c,d$, changing the signs
by bits $\epsilon$ preserves their cancellation exactly when
\begin{equation}
 \epsilon_{rc}+\epsilon_{sc}+\epsilon_{rd}+\epsilon_{sd}=0
 \pmod2.
 \label{eq:sign-parity}
\end{equation}
The equations from all three cuts have rank 28 over $\mathbb F_2$ on
60 variables, leaving a 32-dimensional sign space. Local diagonal sign
changes span an 18-dimensional subspace. These counts describe only
this fixed-support sign problem; they are not counts of LU classes.
The explicit flip used in the main text satisfies every equation.

The invariant calculation also uses integer arithmetic. Enumerating the
four nonzero ket entries fixes all four bra entries in
Eq.~\eqref{eq:klein}. Products are accumulated in $\mathbb Z[\sqrt2]$
and divided by $4^8$. The coefficient of $\sqrt2$ vanishes, leaving
the exact values 108, 72 and 40. If the entry at $0132$ carries a phase
$z$, the same calculation retains its exponent, with opposite sign
for each bra occurrence. For Bruzda's tensor the resulting Laurent
polynomial is exactly
\begin{equation}
 K(T^{\rm WB}(z))=56+8z+8z^{-1},\qquad |z|=1.
\end{equation}
Orthogonality holds for every $|z|=1$, because the marked entry is isolated in its row and column in every flattening.

To see why $K$ is also invariant under party permutations, label the four
copies by $V=\mathbb F_2^2$. The four matchings are
$K_p(r)=r+p$, one for each party label $p\in V$. Every permutation of
these four labels has the affine form $p\mapsto Lp+t$, because
$\operatorname{AGL}(2,2)$ has 24 elements. Relabelling the two sets of
copies independently by $f(r)=Lr+t$ and $g(r)=L^{-1}r$ gives
 $fK_pg=K_{Lp+t}$.
Relabelling dummy indices leaves the contraction unchanged. Thus distinct
values exclude LU equivalence even after arbitrary party permutations.
Finally, the support-component orderings in Fig.~\ref{fig:three-real-blocks} are obtained by an exact graph traversal.
These row and column permutations leave every off-block entry
identically zero. They do not modify the tensor coefficients or claim a
further physical transformation of the state.

For the three real states in Table~\ref{tab:three-real-classes}, direct summation of the Weyl fourth moments in the convention of Eq.~\eqref{eq:magic-M2} gives
\begin{equation}
 \begin{aligned}
 e^{-M_2(\psi_{96})}&=\frac{296851+310\sqrt2}{107495424},\\
 e^{-M_2(\psi_{\rm A})}&=\frac{2117}{559872},\qquad
 e^{-M_2(\psi_{\rm B})}=\frac{2137}{559872},
 \end{aligned}
 \label{eq:magic-exact-moments}
\end{equation}
where $\psi_{\rm A}=T^{\rm WB}/6$ and
$\psi_{\rm B}=T^{{\rm WB},-}/6$. These moments were evaluated
in $\Q(\sqrt2)$ and independently checked by summing the $6^8$
Weyl expectation values. The comparison with the published golden
state uses its original coefficients and the same Weyl convention;
its quoted entropy is evaluated numerically.

\section{On the quantum seven-dice problem}
\label{app:seven-party}

We analyse here the $\AME(7,6)$ state of
Ref.~\cite{Shi2026} and its fractional Latin representation.
At each site write $v_j=3b_j+x_j$ with $b_j\in\mathbb F_2$ and
$x_j\in\mathbb F_3$. Let $C\subset\mathbb F_2^7$ be the span of

\begin{equation}
 H=\begin{pmatrix}
 1&0&1&1&1&0&0\\
 1&1&1&0&0&1&0\\
 0&1&1&1&0&0&1
 \end{pmatrix}.
\end{equation}

The binary states are
$|a_L\rangle=8^{-1/2}\sum_{b\in C+a\mathbf1}|b\rangle$,
$a=0,1$. Put $\omega=e^{2\pi i/3}$ and
\begin{align}
 |G_c\rangle&=3^{-7/2}\sum_{x\in\mathbb F_3^7}
        \omega^{Q(x)+c\cdot x}|x\rangle,\nonumber\\
 Q(x)&=\sum_{i<j}\Gamma_{ij}x_ix_j,\qquad
 \Gamma_{ij}=\begin{cases}
 1,&i-j\notin\{0,1,6\}\pmod7,\\
 0,&\text{otherwise}.
 \end{cases}
\end{align}
The state of Ref.~\cite{Shi2026} used here is
\begin{equation}
 |\Psi\rangle=\frac{|0_L\rangle|G_0\rangle+
                         |1_L\rangle|G_c\rangle}{\sqrt2},
 \qquad c=(1,2,1,0,0,0,0),
 \label{eq:seven-party-source}
\end{equation}
with the binary and ternary factors paired site by site into six-level
systems. There are $16\cdot3^7=34992$ nonzero amplitudes. An exact
check of every three-party Gram matrix counts phase differences
$0,1,2$: off the diagonal their counts agree, giving zero by
$1+\omega+\omega^2=0$; a diagonal row contains 162 terms, giving
$162/(16\cdot3^7)=6^{-3}$.

\subsection{Exact character form of the 48 blocks}
Use the ordering of Section~\ref{sec:seven-hadamard}, with
$a=4b_0+2b_1+b_2$, branch $\epsilon\in\{0,1\}$ in
Eq.~\eqref{eq:seven-party-source}, and $s=3\epsilon+x_4$.
Put $x=(x_0,x_1,x_2)$ and $u=(x_3,x_5,x_6)$ in $\F_3^3$.
The global row and column indices of block $(a,s)$ are
\begin{equation}
 r=27a+9x_0+3x_1+x_2,\qquad
 c=162a+27s+9x_3+3x_5+x_6.
\end{equation}
The binary support pairs each input triple with two output quadruples,
one per branch. Allowing the three values of $x_4$ gives the six
fully populated squares in each of the eight rectangles.

In the phase exponent of $H_{a,s}$, the terms involving only $x$ or
only $u$ are removed by row and column phases. The common dephased
matrix is
\begin{equation}
 \widehat H_{x,u}=\omega^{x^{\mathsf T}Mu},\qquad
 M=\begin{pmatrix}1&1&0\\1&1&1\\0&1&1\end{pmatrix},
 \qquad \det M=-1\ne0\quad\hbox{in }\F_3.
 \label{eq:seven-hadamard-kernel}
\end{equation}
For two rows $x,y$,
\begin{equation}
 \sum_{u\in\F_3^3}\omega^{(x-y)^{\mathsf T}Mu}
 =27\,\delta_{x,y}.
\end{equation}
Indeed, for $x\ne y$ the nonzero linear character takes each value
$0,1,2$ exactly nine times. Since $x\mapsto M^{\mathsf T}x$ is
bijective, $\widehat H/\sqrt{27}$ is $F_3^{\otimes3}$ after a row permutation.
Thus $H_{a,s}H_{a,s}^{\dagger}=27I_{27}$ for every $(a,s)$.
All 48 blocks have off-diagonal phase counts $(9,9,9)$ and the same
dephased kernel. Their $34992$ entries exhaust the support of the
full coefficient matrix.

\section{The quantum nine-dice problem and six Hadamard blocks}
\label{app:nine-phase}
We investigate the non-stabilizer
$\AME(9,6)$ state constructed by Kitamura~\cite{Kitamura2026}
at commit \texttt{500ee827dffb}
and look for its visualization.

We use the integer matrices $G,C,B,H$ from the version of
Ref.~\cite{Kitamura2026} specified above. Write $z_i=3x_i+y_i$, with
$x_i\in\{0,1\}$ and $y_i\in\mathbb F_3$. For a symmetric matrix $M$
with zero diagonal, let $q_M(u)=\sum_{i<j}M_{ij}u_i u_j$.
With $\eta=\exp(\pi i/3)$, the amplitudes are
\begin{equation}
\Psi(z):=\langle z_0,\ldots,z_8|\Psi_9\rangle
 =6^{-9/2}\eta^{f(z)},\qquad
 f(z)=3q_G(x)+2\bigl(q_C(y)+x^{\mathsf T}By+q_H(x)\bigr)\pmod6.
 \label{eq:ame96-phase}
\end{equation}
The quadratic form $q_G$ is binary; the expression in parentheses is
reduced modulo three. In particular, the binary entries of $x$ in the
mixed and $H$ terms are embedded as the integers $0,1$ in $\mathbb F_3$.
They must not be replaced by binary differences when evaluating
ternary phases.

Let $S=(0,1,2,3)$, $P=(4,5,6,8)$ and fix $z_7=t$. The relevant blocks are
\begin{equation}
 G_{SP}=\begin{pmatrix}
 1&1&1&1\\0&1&0&1\\0&1&1&0\\1&0&1&1
 \end{pmatrix},\quad
 C_{SP}=\begin{pmatrix}
 0&2&1&1\\1&2&0&0\\1&2&0&1\\0&1&0&1
 \end{pmatrix},\quad
 B_{SP}=\begin{pmatrix}
 2&1&0&2\\0&1&1&1\\0&1&0&1\\2&1&2&0
 \end{pmatrix}.
 \label{eq:ame96-crossblocks}
\end{equation}
Here $G_{SP}$ is invertible over $\mathbb F_2$ and $C_{SP}$ is
invertible over $\mathbb F_3$. The source also has
$B_{PS}=0$ and $H_{SP}=0$. Consequently, after collecting all terms
depending only on the row or only on the column into $r_t$ and $c_t$,
the unimodular block $E_t=6^2V_t$ reads
\begin{equation}
 (E_t)_{(x,y),(u,v)}
 =\eta^{r_t(x,y)+c_t(u,v)}
 (-1)^{x^{\mathsf T}G_{SP}u}
 \omega^{(y^{\mathsf T}C_{SP}+x^{\mathsf T}B_{SP})v},
 \qquad \omega=\eta^2.
 \label{eq:ame96-character-block}
\end{equation}
Both $x,u$ range over $\mathbb F_2^4$ and both $y,v$ over
$\mathbb F_3^4$. The map
\begin{equation}
 (x,y)\longmapsto
 \bigl(x^{\mathsf T}G_{SP},\;
 y^{\mathsf T}C_{SP}+x^{\mathsf T}B_{SP}\bigr)
 \label{eq:ame96-character-map}
\end{equation}
is bijective: the first component determines $x$, after which the
second determines $y$. Removing the row and column phases in
Eq.~\eqref{eq:ame96-character-block} therefore leaves precisely a
character table, up to row and column ordering. Orthogonality follows
from
\begin{equation}
 \sum_{u\in\mathbb F_2^4}(-1)^{r\cdot u}
 \sum_{v\in\mathbb F_3^4}\omega^{s\cdot v}
 =6^4\,\delta_{r,0}\delta_{s,0}.
\end{equation}
This proves all six Hadamard claims, including the mixed term in $B$.
An explicit Fourier labelling is obtained by writing
$r=x^{\mathsf T}G_{SP}$ and $s=y^{\mathsf T}C_{SP}+x^{\mathsf T}B_{SP}$
in their respective fields and assigning row labels $3r+2s\pmod6$
and column labels $3u+4v\pmod6$. Their ordinary scalar product modulo six
is exactly $3r\cdot u+2s\cdot v$.

The order-$36$ matrices $K^{a,b,c,d,t}$ in
Eq.~\eqref{eq:ame96-submatrix36} admit the same argument on two
coordinates. Write $(u,v)=3x+y$ and $(w,s)=3p+q$ componentwise,
with $x,p\in\mathbb F_2^2$ and $y,q\in\mathbb F_3^2$. The source
cross matrices for row parties $(2,3)$ and column parties $(5,8)$ are
\begin{equation}
 G_*=I_2,\qquad
 C_*=\begin{pmatrix}2&1\\1&1\end{pmatrix},\qquad
 B_*=\begin{pmatrix}1&1\\1&0\end{pmatrix},
 \label{eq:ame96-small-crossblocks}
\end{equation}
and the reverse mixed and $H$ cross terms vanish. Dephasing any
$K^{a,b,c,d,t}$ therefore gives the same kernel
\begin{equation}
 \widehat K_{(x,y),(p,q)}
 =(-1)^{x^{\mathsf T}p}
  \omega^{(y^{\mathsf T}C_*+x^{\mathsf T}B_*)q}.
 \label{eq:ame96-small-kernel}
\end{equation}
All dependence on the fixed labels $a,b,c,d,t$ is in row and column
phases. Since $\det C_*=1$ in $\mathbb F_3$, the character map
$(x,y)\mapsto(x^{\mathsf T},y^{\mathsf T}C_*+x^{\mathsf T}B_*)$
is bijective. Character orthogonality then gives
$\widehat K\widehat K^\dagger=36I_{36}$ and hence
$K^{a,b,c,d,t}(K^{a,b,c,d,t})^\dagger=36I_{36}$ for all fixed labels.
For each $t$, the $36$ row groups indexed by $(a,b)$ and the $36$
column groups indexed by $(c,d)$ partition $E_t$ into $6^4$ such
locations. Across all six full blocks this gives $6^5$ disjoint
order-$36$ Hadamard submatrices, covering every coefficient once:
$6^5\cdot36^2=6^9$.

\section{Golden \texorpdfstring{$AME(4,6)$}{AME(4,6)} state
and fractional Latin square}
\label{app:golden-officers}
The chess representation in Fig.~6 of Ref.~\cite{Zyczkowski2023}
already records both output labels, the coefficient moduli and their
relative phases. Figure~\ref{fig:golden-officers} displays the same golden
construction as a partition into probability areas. Each normalized
two-party officer occupies a unit square. A region of area $p$ represents
a coefficient of modulus $\sqrt p$ in that officer and $\sqrt p/6$ in
the four-party state. The difference is thus in the representation,
rather than in the underlying state: the areas now obey a common,
quantitative rule and fill each cell completely.
Let $|\psi_{\rm gold}\rangle$ denote the normalized golden state,
and define its officer coefficients by
\begin{equation}
 T^{\rm gold}_{ijkl}
 :=6\langle i,j,k,l|\psi_{\rm gold}\rangle.
\end{equation}
Then the table inputs are 
\begin{equation}
 |\Phi^{\rm gold}_{ij}\rangle
 =\sum_{k,l}T^{\rm gold}_{ijkl}|k,l\rangle,\qquad
 |\psi_{\rm gold}\rangle
 =\frac16\sum_{i,j}|i,j\rangle|\Phi^{\rm gold}_{ij}\rangle,
\end{equation}
and every $|\Phi^{\rm gold}_{ij}\rangle$ is normalized. Summing its
probabilities over either output label gives
\begin{equation}
 P^{\rm gold}_{ijk}=\sum_l|T^{\rm gold}_{ijkl}|^2,\qquad
 Q^{\rm gold}_{ijl}=\sum_k|T^{\rm gold}_{ijkl}|^2.
 \label{eq:golden-fractional}
\end{equation}

Both arrays satisfy the three line-sum conditions of
Eq.~\eqref{eq:tristochastic}. For example, normalization of an officer
gives $\sum_l Q^{\rm gold}_{ijl}=1$, while the diagonal parts of the
other two unitarity relations give
$\sum_j Q^{\rm gold}_{ijl}=\sum_i Q^{\rm gold}_{ijl}=1$.
The same argument applies to $P^{\rm gold}$.
In the figure, colour labels $l$, so its total area in a cell is
$Q^{\rm gold}_{ijl}$; each colour therefore occupies total area one
along every row and column. The exact coefficient data verify all
108 line sums for each array.

These are fractional Latin squares in the real-valued sense used here,
but not rational ones: their entries include $(5\pm\sqrt5)/20$.
The line sums concern total areas. They do not, by themselves, imply
that the displayed fragments can be overlaid at fixed positions to
fill a monochromatic square. Finally, unlike the real officer array,
the golden square requires phases beyond a sign. We retain them as
powers of $\omega=e^{i\pi/10}$, with the index and phase convention
used in the figure.

\begin{figure}[p]
 \centering
 \includegraphics[width=0.95\linewidth]{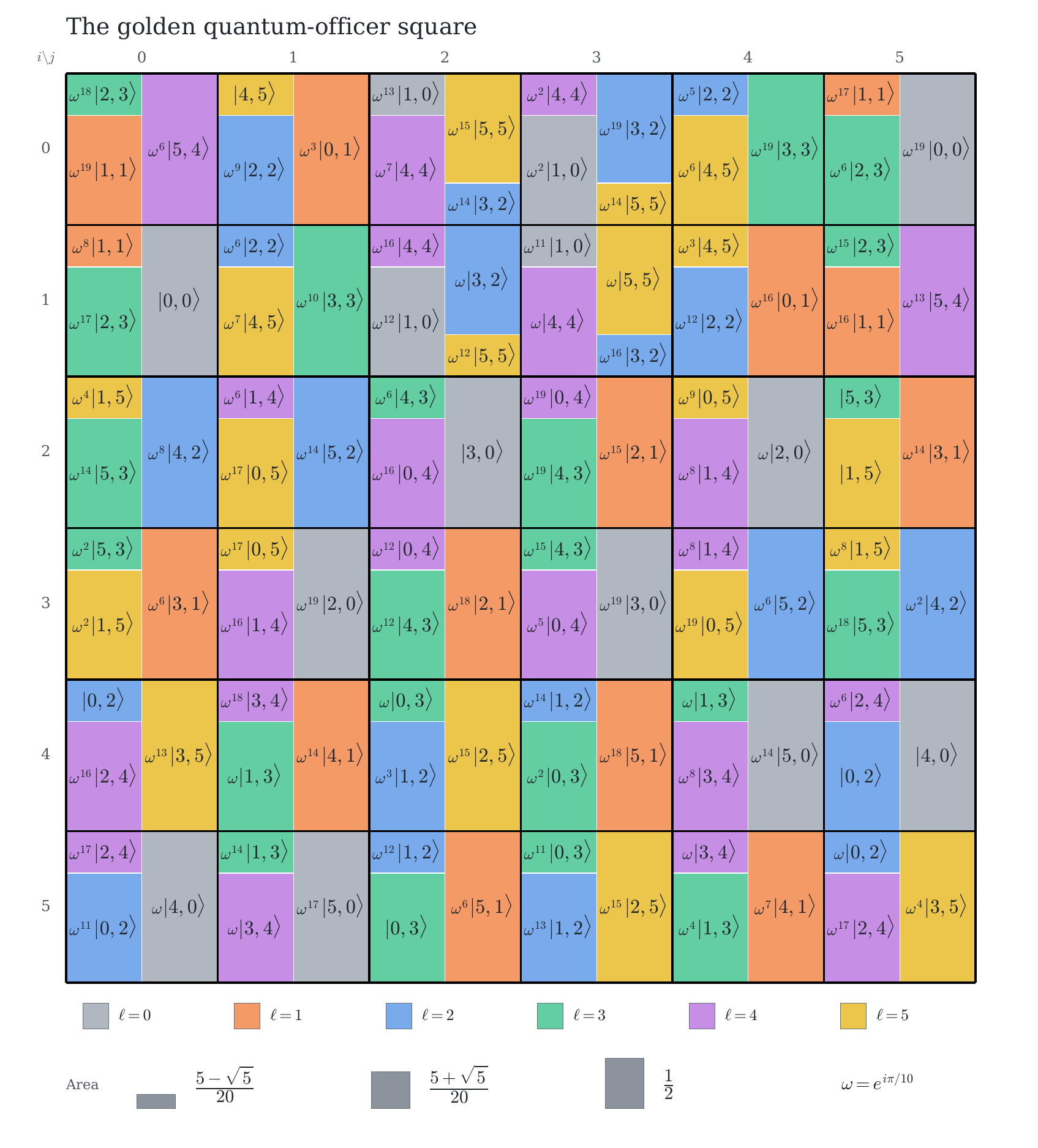}
 \caption{The 36 officer cells of the golden $\AME(4,6)$
 state~\cite{Rather2022,Zyczkowski2023}
 and the corresponding
fractional Latin square of order six.
  Black grid lines separate entangled
 officers; the finer internal divisions separate their product-basis
 terms $|k,l\rangle$. Colour labels the second output $l$, and the
 area of each rectangle is the squared modulus of its coefficient.
 The three possible areas are $(5-\sqrt5)/20$, $(5+\sqrt5)/20$ and $1/2$.
 The displayed powers of $\omega=e^{i\pi/10}$ give the relative phases;
 a missing phase factor means one. All 112 nonzero terms are shown.
 Summing the areas of each colour gives one in every row and column,
 as expressed by Eq.~\eqref{eq:golden-fractional}. The phase labels
 retain the coherent information absent from those area sums.}
 \label{fig:golden-officers}
\end{figure}

\end{document}